\documentclass{article}

\usepackage{microtype}
\usepackage{graphicx}
\usepackage{subcaption}
\usepackage{multicol}
\usepackage{booktabs} % for professional tables

\usepackage{hyperref}

\usepackage
{
        amssymb,
        amsmath,
        amsthm,
}
\usepackage{caption}
\usepackage[ruled,vlined,linesnumbered]{algorithm2e}
\usepackage{tikz}

\usepackage{ifthen}

\newtheorem{theorem}{Theorem}[section]
\newtheorem{lemma}[theorem]{Lemma}
\newtheorem{proposition}[theorem]{Proposition}

\newtheorem{corollary}[theorem]{Corollary}
\newtheorem{definition}[theorem]{Definition}
\newtheorem{remark}[theorem]{Remark}

\DeclareMathOperator {\diam}  {diam}

\DeclareMathOperator {\cost}  {cost}

\newcommand {\E}       {\mathbb{E}}

\newcommand {\bbR}    {\mathbb{R}}

\newcommand {\F}    {\mathcal{F}}

\newcommand{\eps}{\varepsilon}

\makeatletter
\newcommand{\removelatexerror}{\let\@latex@error\@gobble}
\makeatother

\usepackage[accepted]{style/icml2021}

\icmltitlerunning{Randomized Dimensionality Reduction for Clustering}

\begin{document}
\newcounter{num}
\setcounter{num}{2}

\twocolumn[
\icmltitle{Randomized Dimensionality Reduction for Facility Location and Single-Linkage Clustering}



\icmlsetsymbol{equal}{*}

\begin{icmlauthorlist}
\icmlauthor{Shyam Narayanan}{equal,mit}
\icmlauthor{Sandeep Silwal}{equal,mit}
\icmlauthor{Piotr Indyk}{mit}
\icmlauthor{Or Zamir}{ias}
\end{icmlauthorlist}

\icmlaffiliation{mit}{Electrical Engineering and Computer Science Department, Massachusetts Institute of Technology, Cambridge, MA, USA}
\icmlaffiliation{ias}{Institute of Advanced Study, Princeton, NJ, USA}

\icmlcorrespondingauthor{Shyam Narayanan}{shyamsn@mit.edu}
\icmlcorrespondingauthor{Sandeep Silwal}{silwal@mit.edu}
\icmlcorrespondingauthor{Piotr Indyk}{indyk@mit.edu}
\icmlcorrespondingauthor{Or Zamir}{orzamir@ias.edu}

\icmlkeywords{Dimension Reduction, Clustering, Facility Location, MST, Doubling Dimension}

\vskip 0.3in
]



\printAffiliationsAndNotice{\icmlEqualContribution} % otherwise use the standard text.

\begin{abstract}
Random dimensionality reduction is a versatile tool for speeding up algorithms for high-dimensional problems. We study its application to two clustering problems: the facility location problem, and the single-linkage hierarchical clustering problem, which is equivalent to computing the minimum spanning tree. We show that if we project the input pointset $X$ onto a random $d = O(d_X)$-dimensional subspace (where $d_X$ is the doubling dimension of $X$),
%\todo{technically $\log \lambda_X$ is: either replace with $\log \lambda$ or use $d_X$}), 
then the optimum facility location cost in the projected space approximates the original cost up to a constant factor. We show an analogous statement for minimum spanning tree, but with the dimension $d$ having an extra $\log \log n$ term and the approximation factor being arbitrarily close to $1$. Furthermore, we extend these results to approximating {\em solutions} instead of just their {\em costs}. Lastly, we provide experimental results to validate the quality of solutions and the speedup due to the dimensionality reduction.
Unlike several previous papers studying this approach in the context of $k$-means and $k$-medians, our dimension bound does not depend on the number of clusters but only on the intrinsic dimensionality of $X$. 
\end{abstract}

\section{Introduction}  \label{sec:intro}
Clustering is a fundamental problem with many applications in machine learning, statistics, and data analysis. Although many formulations of clustering  are NP-hard in the worst case, many heuristics and approximation algorithms exist and are widely deployed in practice. Unfortunately, many of those algorithms suffer from large running times, especially if the input data sets are high-dimensional.

In order to improve the performance of clustering algorithms in high-dimensional spaces, a popular  approach is to project the input point set into a lower-dimensional space and perform the clustering in the projected space. Reducing the dimension (say, from $m$ to $d \ll m$) has multiple practical and theoretical advantages, including (i)  lower storage space, which is linear in $d$ as opposed to $m$; (ii) lower running time of the clustering procedure - the running times are often dominated by distance computations, which take time linear in the dimension; and (iii) versatility: one can use {\em any}  algorithm or its implementation to cluster the data in the reduced dimension. 
Because of its numerous benefits, dimensionality reduction as a tool for improving algorithm performance has been studied extensively, leading to many theoretical tradeoffs between the projected dimension and the solution quality.  A classic result in this area is the Johnson-Lindenstrauss (JL) lemma \yrcite{originalJL} which (roughly) states that a random projection of a dataset $X \subseteq \mathbb{R}^m$ of size $n$ onto a dimension of size $O(\log n)$ approximately preserves all pairwise distances. This tool has been subsequently applied to many clustering and other problems (see \cite{dim_reductionsurvey} and references therein).

Although the JL lemma is known to be tight~\cite{larsen2017optimality} in general, better tradeoffs are possible for {\em specific} clustering problems. Over the last few years, several works \cite{kmeans1, kmeans2, becchetti2019oblivious, ilyapaper} have shown that combining random dimensionality reduction with $k$-means leads to better guarantees than implied by the JL lemma. In particular, a recent paper by Makarychev, Makarychev, and Razenshteyn \yrcite{ilyapaper} shows that to preserve the $k$-means cost up to an arbitrary accuracy,  it suffices to project the input set $X$ onto a dimension of size $O(\log k)$, as opposed to $O(\log n)$ guaranteed by the JL lemma. Since $k$ can be much smaller than $n$, the improvement to the dimension bound can be substantial.
However, when $k$ is comparable to $n$, the improvement is limited. This issue is particularly salient for clustering problems with a variable  number of clusters, where no a priori bound on the number of clusters exists.

In this paper we study randomized dimensionality reduction over Euclidean space $\mathbb{R}^m$ in the context of two fundamental clustering problems with a variable number of clusters. In particular:
\begin{itemize}
    \item Facility location (FL): given a set of points $X \subset \mathbb{R}^m$ and a facility opening cost, the goal is to open a subset $\mathcal{F} \subseteq X$ of facilities in order to minimize the total cost of opening the facilities plus the sum of distances from points in $X$ to their nearest facilities (see Section \ref{sec:prelim} for a formal definition). Such cost functions are often used when the ``true'' number of clusters $k$ is not known, see e.g., \cite{cambridge2009online}, section 16.4.1.
    \item Single-linkage clustering, or (equivalently) Minimum Spanning Tree (MST): given a set of points $X \subset \mathbb{R}^m$, the goal is to connect them into a tree in order to minimize  the total cost of the tree edges. This is a popular variant of Hierarchical Agglomerative Clustering (HAC) that creates a hierarchy of clusters, see e.g.,  \cite{cambridge2009online}, section 17.2.
\end{itemize}

We remark that some papers, e.g., \cite{abboud2019hac} define approximate HAC operationally, by postulating that each step of the clustering algorithm must be approximately correct. However, there are other theoretical formulations of approximate HAC as well, e.g., \cite{dasgupta2016hac, moseley2017hac}. 
Since single-linkage clustering has a natural objective function induced by MST, defining approximate single-linkage clustering as approximate MST is a natural, even if not unique, choice.

\paragraph{Our Results} Our main results show that, for both FL and MST, it is possible to project input point sets into low (sometimes even constant) dimension while provably preserving the quality of the solutions. Specifically, our theorems incorporate the \emph{doubling dimension} $d_X$ of the input datasets $X$. This parameter\footnote{We formally define it in Section \ref{sec:prelim}.} measures the ``intrinsic dimensionality'' of $X$ and can be much lower than its ambient dimension $m$. If $X$ has size $n$, the doubling dimension $d_X$ is always at most $\log n$, and is often much smaller. We show that random projections into dimension roughly proportional to $d_X$ suffice in order to approximately preserve the solution quality. The specific bounds are listed in Table \ref{table}. 

\begin{figure*}[t]
\begin{center}
\caption{Number of dimensions $d$ required for a random projection to provide a good clustering approximation}
\vskip 0.1in
\begin{tabular}{|l|l|l|l|l|}
\hline
Problem & Proj. dimension $d$ & Approx. & Cost/Solution & Reference \\
\hline
FL & $O(d_X)$ & $O(1)$ & Cost & Theorem \ref{thm:constant}\\
FL & $O(d_X)$ & $O(1)$ & Locally optimal solution & Theorem \ref{cor:main}\\
\hline
MST & $O(1/\epsilon^2 \cdot (d_X \log(1/\epsilon) +\log \log n))$ & $1+\epsilon$ & Cost  & Theorem~\ref{thm:MainMST}\\
MST & $O(1/\epsilon^2 \cdot (d_X \log(1/\epsilon) +\log \log n))$ & $1+\epsilon$ & Optimal solution & Theorem~\ref{thm:MainMST}\\
\hline
\end{tabular}
\label{table}
\end{center}
\vskip -0.1in
\end{figure*}

We distinguish between two types of guarantees. The first type states that the minimum {\em cost} of FL or MST is preserved by a random projection (with high probability) up to the specified factor. This guarantee is useful if the goal is to quickly estimate the optimal value. The second type states that a {\em solution} computed in the projected space induces a solution in the original space which approximates the best solution (in the original space) up to the specified approximation factor. This guarantee implies that one can find an approximately optimal clustering  by mapping the data into low dimensions and clustering the projected data. To obtain the second guarantee, we need to assume that the solution in the projected space is either globally optimal (for MST) or locally optimal\footnote{Informally, a solution is locally optimal if opening any new facility does not decrease its cost. The formal definition is slightly more general, and is given in Section~\ref{sec:local}. Note that any solution found by local search algorithms such as that in~\cite{MP_alg} satisfies this  condition.} (for FL). We note that these two types of guarantees are  incomparable. In fact, for FL, our proofs of the cost and of the solution guarantees are substantially different. We also prove analogous theorems for the ``squared'' version of FL, where the distance between points is defined as the {\em square} of the Euclidean distance between them.

We complement the above results by showing that the conditions and assumptions in our theorem cannot be substantially reduced or eliminated. Specifically, for both FL and MST, we show that:
\begin{itemize}
    \item The bounds on the projected dimension $d$ in the theorems specified in the table must be at least $\Omega(d_X)$, as otherwise the approximation factors for both the cost and the solution become super-constant (Theorems \ref{thm:lb_fac}, \ref{thm:lb_mst_cost}, \ref{thm:lb_mst_pullback})
    \item The assumptions that the solution in the projected space is (locally) optimal cannot be relaxed to ``approximately optimal'' (Lemmas \ref{lem:opt_nec}, \ref{lem:opt_nec_mst}).
\end{itemize}
Also, we show that, in contrast to facility location and MST, one must project to $\Omega(\log k)$ dimensions for preserving both the cost and solution for $k$-means and $k$-medians clustering, even if the doubling dimension $d_X$ is $O(1)$.

Finally, we present an experimental evaluation of the algorithms suggested by our results. Specifically, we show that both FL and MST, solving these problems in reduced dimension can reduce the running time by 1-2 orders of magnitude while increasing the solution cost only slightly. We also give empirical evidence that the doubling dimension of the input point set affects the quality of the approximate solutions. Specifically, we study two simple point sets of size $n$ that have  similar structure but very different doubling dimension values ($O(1)$ and $O(\log n)$, respectively). We empirically show that a good approximation of the MST can be found for the former point set by projecting it into much fewer dimensions than the latter point set.

\paragraph{Related Work}
There is a long line of existing work on approximating the solution of various clustering problems in metric spaces with small doubling dimensions (see \cite{db_dim1, db_dim2, db_dim3, db_dim4}). The state of the art result is given in \cite{focspaper} where a near linear $(1+\epsilon)$-approximation algorithm is given for a variety of clustering problems. However, these runtimes have a doubly-exponential dependence on $d$ which is proven to be unavoidable unless P = NP \cite{focspaper}. 
For MST in spaces of doubling dimension $d_X$, it is known that an $(1+\epsilon)$-approximate solution can be computed in time $2^{O(d_X)} n \log n + \epsilon^{-O(d_X)} n$ \cite{gottlieb2013proximity}.
To the best of our knowledge, none of these algorithms have  been implemented.

In addition, the notion of doubling dimension has also been previously used to study algorithms for high dimensional problems such as the nearest neighbor search, see e.g., \cite{indyknaor, nn1, nn2}. The paper~\cite{indyknaor} is  closest in spirit to our work, as it shows that, for a fixed point $q$ and a data set $X$,  a random projection into $O(d_X)$ dimensions approximately preserves the distance from $q$ to its nearest neighbor in $X$ with a ``good'' probability. If the probability of success was of the form $1-1/{2n}$, we could apply this statement to all (up to $n$) facilities in the solution simultaneously, which would prove our results.  Unfortunately, the probability of failure is much higher than $1/n$, and therefore this approach fails. Nevertheless, our proofs use some of the lemmas developed in that work, as discussed in Section~\ref{sec:prelim}. 

\section{Preliminaries} \label{sec:prelim}
\paragraph{Problem Definitions}
The \emph{Euclidean Facility Location} problem is defined as follows: We are given a dataset $X \subset \mathbb{R}^m$ of $n$ points and a nonnegative function $c: X \rightarrow \mathbb{R}$ that represents the cost of \emph{opening} a facility at a particular point. The goal is to find a subset $\mathcal{F} \subseteq X$ that minimizes the objective $\cost(\mathcal{F}) =  \sum_{f \in \mathcal{F}} c(f) + \sum_{x \in X} D(x, \mathcal{F}),$ where $D(x, \mathcal{F}) = \min_{f \in \mathcal{F}} \|x-f\|$. In this work we restrict our attention to the case that $\| \cdot \|$ is the Euclidean $(\ell_2)$ metric. The first term $\sum_{f \in \mathcal{F}} c(f)$ is referred to as the opening costs and the second term $ \sum_{x \in X} D(x, \mathcal{F})$ is referred to as the connection costs. In this work, we also focus on the \emph{uniform} version of facility location where 
%we assume that 
all opening costs are the same. By re-scaling the points, we can further assume that $f(x) = 1$ for all $x \in X$. Therefore, throughout the paper, we focus on minimizing the following objective function:
\begin{equation}\label{eq:objective}
\cost(\mathcal{F}) =| \mathcal{F}| + \sum_{x \in X} \,  \min_{f \in \mathcal{F}}\|x-f\|. 
\end{equation}
A set $\mathcal{F}$ of facilities is also referred to as a \emph{solution} to the facility location problem.

The \emph{Euclidean Minimum Spanning Tree} problem is defined as follows. Given a dataset $X \subset \mathbb{R}^m$ of $n$ points, we wish to find a set $\mathcal{M}$ of edges $(x, y)$ that forms a spanning tree of $X$ and minimizes the following objective function:
\begin{equation}\label{eq:objective_mst}
\cost(\mathcal{M}) = \sum_{(x, y) \in \mathcal{M}}\|x-y\|.
\end{equation}

\paragraph{Properties of Doubling Dimension}
We parameterize our dimensionality reduction using \emph{doubling dimension}, a measure of the intrinsic dimensionality of the dataset. The notion of doubling dimension holds for a general metric space $X$ and is defined as follows. Let $B(x,r)$ denote the ball of radius $r$ centered at $x \in X$, intersected with the points in $X$. Then the doubling \emph{constant} $\lambda_X$ is the smallest constant $\lambda$ such that for all $x \in X$ and for all $r > 0$, there exists $S\subseteq X$ with $|S| \le \lambda$ such that $B(x,r) \subseteq \bigcup_{s \in S} B(s, r/2)$. The doubling \emph{dimension} of $X$ is is defined as $d_X := \log \lambda_X$. One can see that $\lambda_X \le |X|,$ so $d_X \le \log |X|$.
In this paper, we focus on the case that $X$ is a subset of Euclidean space $\mathbb{R}^m$. 
\paragraph{Dimension Reduction}
In this paper we define a random projection as follows.

\begin{definition}\label{def:proj}
 A random projection from $\mathbb{R}^m$ to $\mathbb{R}^d$ is a linear map $G$ with i.i.d.\@ entries drawn from $\mathcal{N}(0, 1/d)$.
\end{definition} 

The following dimensionality reduction result related to doubling dimension was proven in \cite{indyknaor}. Informally, the lemma below states that a random projection of $X$ onto a dimension $O(d_X)$ subspace does not `expand' $X$ very much.

\begin{lemma}[Lemma $4.2$ in \cite{indyknaor}]\label{lem:indyk} Let $X \subseteq B(0,1)$ be a subset of the $m$-dimensional Euclidean unit ball, and let $G$ be a random projection from $\mathbb{R}^m$ to $\mathbb{R}^d$. Then there exist universal constants $c, C > 0$ such that for $d \ge C \cdot d_X + 1$ and $t > 2$, $\Pr(\exists x \in X, \, \|Gx\| \ge t) \le \exp(-cdt^2)$.
\end{lemma}

For our proofs, we will need some additional preliminary results on random projections, which are deferred to Appendix \ref{sec:app_prelim}.

\section{Local Optimality for Facility Location}
\label{sec:local}
We now define the notion of a locally optimal solution for facility location. As stated in the introduction, this notion plays a key role in our approximation guarantees. Before we present our criterion for local optimality, we begin by discussing the Mettu Plaxton (MP) algorithm, an approximation algorithm for the facility location problem. The MP approximation algorithm gives a useful geometric quantity to understand the facility location problem.

\subsection{Approximating the Cost of Facility Location}
For each $p \in X$, we associate with it a radius $r_p > 0$ which satisfies the relation
%\ifthenelse{\value{num}=1}{
\begin{equation}\label{eq:rdef}
    \sum_{q \in B(p, r_p) \cap X} (r_p - \|p-q\|) = 1.
\end{equation}
%}
It can be checked that a unique value $r_p$ satisfying $1/n \le r_p \le 1$ exists for every $p$. The geometric interpretation of $r_p$ is shown in Figure $\ref{fig:circle_radii}$. This quantity was first defined by Mettu and Plaxton \yrcite{MP_alg}, 
%in the context of their approximation algorithm for the facility location problem. They 
who proved that a simple greedy procedure of iteratively selecting facilities that lie in balls of radii $2r_p$ gives a $3$ factor approximation algorithm for the facility location problem. 
For completeness, their algorithm is given in Appendix \ref{sec:app_local}.

    \begin{figure}[h!]
        \centering
      \includegraphics[width=0.35\linewidth]{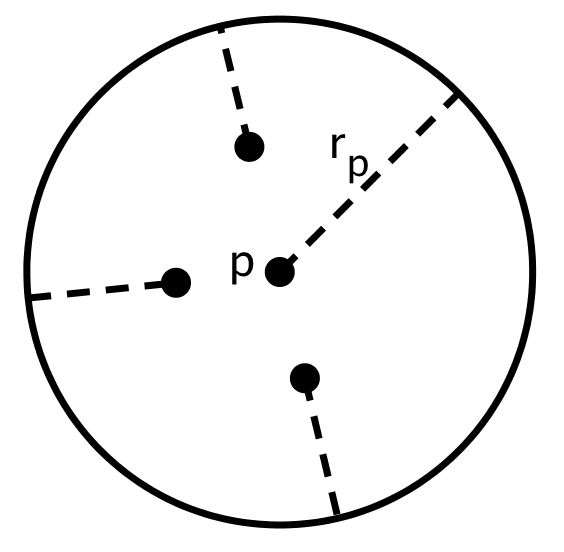}
      \caption{$r_p$ is defined such that the dotted lines add to $1$.}
      \label{fig:circle_radii}
    \end{figure}
%}
One of the main insights from 
%\ifthenelse{\value{num}=1}{
Mettu and Plaxton's algorithm
%}
%{Algorithm \ref{alg:MPalg}}
is that the sum of the radii $r_p$ is a constant factor approximation to the cost of the optimal solution. This insight was first stated in \cite{sublin_MP} where it was used to design a sublinear time algorithm to approximate the cost of the facility location problem. In particular, we have the following result from \cite{sublin_MP} about the approximation property of the radii values.
\begin{lemma}[Lemma $2$ in \cite{sublin_MP}]\label{lem:appx}
Let $C_{OPT}$ denote the cost of the optimal facility location solution. Then 
%\ifthenelse{\value{num}=1}{
$\frac{1}{4} \cdot C_{OPT} \le \sum_{p \in X} r_p \le 6 \cdot C_{OPT}$.
%}
%{$$\frac{1}{4} \cdot C_{OPT} \le \sum_{p \in X} r_p \le 6 \cdot C_{OPT}.$$}
\end{lemma}

For our purposes, we use the radii values to define a local optimality criterion for a solution to the facility location problem. Our local optimality criterion states that each point $p$ must have a facility that is within distance $3r_p$. 

\begin{definition}\label{def:local}
A solution $\mathcal{F}$ to the facility location problem is \emph{locally optimal} if for all $p \in X$, 
%we have 
$B(p, 3r_p) \cap \mathcal{F} \ne \emptyset$.
\end{definition}

%\ifthenelse{\value{num}=1}{
We show in Lemma \ref{lem:localopt} that a solution that is not locally optimal can be improved, i.e. the objective function given in Eq.\@ \eqref{eq:objective} can be improved, by adding $p$ to the set of facilities. This implies that any global optimal solution must also be locally optimal, so requiring a solution of the facility location problem to be locally optimal is a \emph{less restrictive} condition than requiring a solution to be globally optimal.
%}

\begin{lemma} \label{lem:localopt}
Let $\mathcal{F}$ be any collection of facilities. If there exists a $p \in X$ such that $B(p, 3r_p) \cap \mathcal{F} = \emptyset$ then $\textup{cost}(\mathcal{F} \cup \{p\}) < \textup{cost}(\mathcal{F})$,
i.e., we can improve the solution.
\end{lemma}

The proof of Lemma \ref{lem:localopt} is deferred to Appendix \ref{sec:app_local}.

\section{Dimension Reduction for Facility Location}
\subsection{Approximating the Optimal Facility Location Cost}\label{subsec:approx}
In this subsection we show that we can \emph{estimate} the cost of the global optimal solution for a point set $X$ by computing the value of the radii after a random projection onto dimension $d = O(d_X)$. We do this by showing that for each $p$, the value of $r_p$ can be approximated up to a constant multiplicative factor in $\mathbb{R}^d$, the lower dimension.

For each $p \in X$, let $r_p$ and $\tilde{r}_p$ be the radius of $p$ and $Gp$ in $\mathbb{R}^m$ and $\mathbb{R}^d$, respectively, computed according to Eq.\@ \eqref{eq:rdef}. Then we prove that $\E[\tilde{r}_p] = \Theta(r_p)$, where the expectation is over the randomness of the projection $G$.

This proof can be divided into showing $\E[\tilde{r}_p] = O(r_p)$ and $\E[\tilde{r}_p] = \Omega(r_p).$ 
%below copied from Sandeep
Our proof strategy for the former is to use the concentration inequality in Lemma \ref{lem:indyk} to roughly say that points in $B(p, r_p) \cap X$ cannot get `very far' away from $p$ after a random projection. In particular, they must all still be at a distance $O(r_p)$ of $p$ after the random projection. Then using the geometric definition of $r_p$ given in \eqref{eq:rdef} and Figure \ref{fig:circle_radii}, we can say that the corresponding radii of $Gp$ in $\mathbb{R}^d$ denoted as $\tilde{r}_p$ must then be upper bounded by $O(r_p)$.
%below also copied from Sandeep
Our proof strategy for the latter is different in that our challenge is to show that points do not `collapse' closer together. In more detail, for a fixed point $p$, we need to show that after a dimension reduction, many \emph{new} points do not come inside a ball of radius $O(r_p)$ around the point $Gp$. An application of Theorem \ref{thm:indyk} in Appendix \ref{sec:app_prelim},
due to Indyk and Naor \yrcite{indyknaor}, deals with this event.

By adding these expectations over each point $p$ and applying Lemma \ref{lem:appx}, we can prove that the facility location cost is preserved under a random projection. Formally, we obtain the following theorem:

\begin{theorem}\label{thm:constant}
Let $X \subseteq \mathbb{R}^m$ and let $G$ be a random projection from $\mathbb{R}^m$ to $\mathbb{R}^d$ for $d = O(d_X)$. Let $\mathcal{F}_m$ be the optimal solution in $\mathbb{R}^m$ and let $\mathcal{F}_d$ be the optimal solution for the dataset $GX \subseteq \mathbb{R}^d$. Then there exist constants $c, C >0$ such that
$ c \cdot \textup{cost}(\mathcal{F}_m) \le \E[\textup{cost}(\mathcal{F}_d)] \le C \cdot \textup{cost}(\mathcal{F}_m)$.
\end{theorem}

The full proof of Theorem \ref{thm:constant} and the lemmas bounding $\E[\tilde{r}_p]$ are deferred to Appendix \ref{sec:app_fl}.

\subsection{Obtaining Facility Location Solution in Larger Dimension} \label{subsec:pullback}
As discussed in the introduction, for many applications, it is not enough to be able to approximate the \emph{cost} of the optimal solution, but rather \emph{obtain} a good solution.

In particular, we would like to perform dimensionality reduction on a dataset $X$, use some algorithm to solve facility location, and then have the guarantee that the quality of the solution we found is a good indicator of the quality of the solution in the 
original
%\emph{larger} 
dimension. Furthermore, since optimally solving facility location in the smaller dimension might still be a challenging task, it is desirable to have a guarantee that a good solution (not necessarily the global optimum) will be a good solution in the larger dimension. We show in this section that this is indeed the case for \emph{locally optimal} solutions. 

Specifically, we show that the cost of a locally optimal solution found in $\mathbb{R}^d$ does not increase substantially when evaluated in the larger dimension. More formally, we prove the following theorem:

\begin{theorem}\label{cor:main}
Let $X \subset \mathbb{R}^m$ and $G$ be a random projection from $\mathbb{R}^m$ to $\mathbb{R}^d$ for $d = O(d_X \cdot  \log(1/\epsilon)/\epsilon^2)$. Let $\mathcal{F}_d$ be a locally optimal solution for the dataset $GX$. Then, the cost of $\mathcal{F}_d$ evaluated in $\mathbb{R}^m$, denoted as $\textup{cost}_m(\mathcal{F}_d)$, satisfies
$$ \E[\textup{cost}_m(\mathcal{F}_d)] \le |\mathcal{F}_d| + O\bigg(\sum_{p \in X} r_p\bigg) \le \textup{cost}_d(\F_d) + O(F),$$
where $F$ is the optimal facility location cost of $X$ in $\bbR^m$.
\end{theorem}

To describe the proof intuition, first note that the cost function defined in Eq.\@ \eqref{eq:objective} has two components. One is the number of facilities opened, and the other is the connection cost. The first term is automatically preserved in the larger dimension since the number of facilities stays the same. Therefore, the main technical challenge is to show that if a facility is within distance $O(\tilde{r}_p)$ of a fixed point $p$ in $\mathbb{R}^d$ (note that $\tilde{r}_p$ is calculated according to Eq.\@ \eqref{eq:rdef} in $\mathbb{R}^d$), then the facility must be within distance $O(r_p)$ in $\mathbb{R}^m$, the larger dimension. From here, one can use Lemma \ref{lem:appx} to bound $\sum_{p \in X} r_p$ by $O(F)$, and the simple fact that $|\F_d| \le \cost_d(\F_d)$.

The proof of our main technical challenge relies on the careful balancing of the following two events. First, we control the value of the radius $\tilde{r}_p$ and show that $\tilde{r}_p \approx r_p$. In particular, we show that the probability of $\tilde{r}_p \ge k r_p$ for any constant $k$ is exponentially decreasing in $k$.
Next, we need to bound the probability that a `far' point comes `close' to $p$ after the dimensionality reduction. 
While there exists a known result on this (e.g., Theorem \ref{thm:indyk} in Appendix \ref{sec:app_prelim}),
%While Theorem \ref{thm:indyk} in the Appendix roughly states that `far' points do not come too `close', 
we need a novel, more detailed result to quantify how close far points can come after the dimension reduction. 

To study this in a more refined manner, we bucket the points in $X \setminus \{p\}$ according to their distance from $p$, with the $i$th level representing distance approximately $i$ from $p$.
%The distance spacing between buckets will be a linear scale. 
We show that points in $X\setminus \{p\}$ that are in `level' $i$ do not shrink to a `level' smaller than $O(\sqrt{i})$. Note that we need to control this even across all levels. To do this requires a chaining type argument which crucially depends on the doubling dimension of $X$. Finally, a careful combination of probabilities gives us our result.

The proof of Theorem \ref{cor:main} is deferred to Appendix \ref{sec:app_fl}.

\begin{remark}\label{rem:generalization}
Our proof of Theorem \ref{cor:main} generalizes to the case of \emph{arbitrary} opening costs $c_p$ by changing the definition of $r_p$ to be $ \sum_{q \in B(p, r_p)} (r_p - \|x-q\|) = c_p$.
\end{remark}

\subsection{Facility Location with Squared Costs}

Facility location problem with squared costs is the following variant of facility location. Given a dataset $X \subset \mathbb{R}^m$, our goal is to find a subset $\mathcal{F} \subseteq X$ that minimizes the objective
\begin{equation}%\label{eq:objective2}
\cost(\mathcal{F}) =| \mathcal{F}| + \sum_{x \in X} \,  \min_{f \in \mathcal{F}}\|x-f\|^2. 
\end{equation}
In contrast to \eqref{eq:objective}, we are adding the \emph{squared} distance from each point to its nearest facility in $\mathcal{F}$, rather than just the distance. This is comparable to $k$-means, whereas standard facility location is comparable to $k$-medians.

For the facility location problem with squared costs, we are again able to show that a random projection of $X$ into $O(d_X)$ dimensions preserves the optimal cost up to an $O(1)$ factor, and that any locally optimal solution in the reduced dimension has its cost preserved in the original dimension. The formal statements and proofs are very similar to those of the standard facility location problem, and are deferred to Appendix \ref{sec:fl_squared}.

\section{Dimension Reduction for MST} \label{sec:mst}

In this section we demonstrate the effectiveness of dimensionality reduction for the minimum spanning tree (MST) problem. As in the case of facility location, we show that we can \emph{estimate} the cost of the optimum MST solution by computing the MST in a lower dimension, and that the minimum spanning tree in the lower dimension is an \emph{approximate} solution to the high-dimensional MST problem.

This time our approximations, both to the optimum cost and the optimum solution, can be $(1 + \epsilon)$-approximations for any $\epsilon > 0$, as opposed to the constant factor approximations that we could guarantee for facility location. To formally state our theorem, for some spanning tree $T$ of $X$, let $\cost_X(T)$ be the sum of the lengths of the edges in $T$. Likewise, let $\cost_{GX}(T)$ be the sum of the lengths of the edges in $T$, where distances are measured in the projected tree $GX$. Our main result is the following theorem:

\begin{theorem} \label{thm:MainMST}
    For some positive integers $m, d$, let $X \subset \mathbb{R}^m$ be a point set of size $n$ and let $G: \mathbb{R}^m \to \mathbb{R}^d$ be a random projection. Let $\mathcal{M}$ represent the minimum spanning tree of $X$, with $M = \cost_X(\mathcal{M})$ and $\widetilde{\mathcal{M}}$ represent the minimum spanning tree of $GX$, with $\widetilde{M} = \cost_{GX}(\widetilde{\mathcal{M}})$. Then, for some sufficiently large constant $C_6,$ if $d \ge C_6 \cdot \epsilon^{-2} \cdot (\log \epsilon^{-1} \cdot d_X + \log \log n)$, the following are true:
%\vspace{-2mm}
\begin{enumerate}
    \item The cost of the MST is preserved under projection with probability at least $\frac{9}{10}.$ In other words, $\widetilde{M} \in [1-\epsilon, 1+\epsilon] \cdot M$.
    \item The optimal projected MST $\widetilde{\mathcal{M}}$ is still an approximate MST in the original dimension with probability at least $\frac{9}{10}.$ In other words, $\cost_{X}(\widetilde{\mathcal{M}}) \in [1, 1+\epsilon] \cdot M$.
\end{enumerate}
\end{theorem}

Hence, we obtain a significantly stronger theoretical guarantee for preserving the MST than $d = \Theta(\epsilon^{-2} \log n)$, which is promised by the Johnson-Lindenstrauss Lemma, assuming that $d_X$ and $\epsilon^{-1}$ are constant or very small.

Our main technical result in establishing Theorem \ref{thm:MainMST} is the following crucial lemma, which will in fact allow us to prove both parts of the above theorem simultaneously.

\begin{lemma} \label{lem:MST}
    For all notation as in Theorem \ref{thm:MainMST}, $\E[\cost_{X}(\widetilde{\mathcal{M}}) - \cost_{GX}(\widetilde{\mathcal{M}})] \le O(\epsilon) \cdot M.$
\end{lemma}

The proof strategy for Lemma \ref{lem:MST} involves first dividing the edges of $\widetilde{\mathcal{M}}$ into levels based on their lengths, and bounding the difference between edge lengths (pre- and post- projection) in each level separately. To analyze a level consisting of the edges of length approximately $t$, we first partition the point set $X$ (in the original dimension $\mathbb{R}^m$) into balls $B_1, \dots, B_r$ of radius $c \cdot t$ for a small constant $c$, and show via chaining-type arguments that not too many pairs of balls that were originally far apart come close together after the random projection. Moreover, using Lemma \ref{lem:indyk}, we show that almost all of the balls do not expand by much.

Therefore, there are not many \emph{bad} pairs of balls $(B_i, B_j)$, where $(B_i, B_j)$ is bad if there exists $p \in B_i, q \in B_j$ where $\|p-q\|$ is much bigger than $t$ but $\|Gp-Gq\|$ is approximately $t$. Now, assuming that none of the balls expand by much in the random projection, for any bad pair $(B_i, B_j)$ and edges $(p, q)$ and $(p', q')$ with $p, p' \in B_i$ and $q, q' \in B_j,$ we cannot have both edges in the minimum spanning tree of $GX$. This is because $\|Gp-Gq\|, \|Gp'-Gq'\| \approx t,$ but since $B_i$ and $B_j$ have radius $c \cdot t$ and do not expand by much, we can improve the spanning tree by replacing $(Gp, Gq)$ with either $(Gp, Gp')$ or $(Gq, Gq')$. So, each bad pair can have at most $1$ edge in $\widetilde{\mathcal{M}}$, the MST of $GX$. Overall, in each level, not too many edges in $\widetilde{\mathcal{M}}$ can shrink by much after the projection.

The full proofs of Lemma \ref{lem:MST} and Theorem \ref{thm:MainMST} are given in Appendix \ref{sec:app_mst}.

\section{Lower Bounds for Projection Dimension} \label{sec:lower}

In this section, we state various lower bounds for the projection dimension $d$ for both facility location clustering and minimum spanning tree. We also show that, in contrast to facility location, low doubling dimension does not actually help with dimensionality reduction for $k$-means or $k$-medians clustering.
All proofs are deferred to Appendix \ref{sec:app_lb}.

In all results of this section, we think of $X$ as a point set of size $n$ in Euclidean space $\mathbb{R}^m$, and $G$ as a random projection sending $X$ to $GX \subset \mathbb{R}^d$. In this section, for FL, we always let $\F$ be the optimal set of facilities in $X$, with cost $F$, and $\widetilde{\F}$ be the optimal set of facilities in $GX$, with cost $\widetilde{F}$. We define $\mathcal{M}, M, \widetilde{\mathcal{M}}, \widetilde{M}$ analogously for MST. We use $o(1)$ to denote functions going to $0$ as $n \to \infty,$ and $\omega(1)$ to denote functions going to $\infty$ as $n \to \infty$, where $n = |X|$.

First, we show that the dependence of the projected dimension $d$ on the doubling dimension $d_X$ in Theorems
\ref{thm:constant}, \ref{cor:main}, and \ref{thm:MainMST}
are all required to obtain constant factor approximations for either the cost or the pullback solution. Namely, we show the following three theorems:
%Theorems \ref{thm:lb_fac}, \ref{thm:lb_mst_cost}, and \ref{thm:lb_mst_pullback}, which we now state.

\begin{theorem}[FL] \label{thm:lb_fac}
Let $d = o(\log n)$. There exists $X$ with doubling dimension $\Theta(\log n)$, such that with at least $2/3$ probability over $G: \mathbb{R}^m \to \mathbb{R}^d$, $\widetilde{F} = o(1) \cdot F$. Moreover, with probability at least $2/3,$ $\widetilde{\F}$, when pulled back to $X$, has cost $\omega(1) \cdot F$ in the original dimension.
\end{theorem}

\begin{theorem}[MST] \label{thm:lb_mst_cost}
    Let $d = o(\log n).$ There exists $X$ with doubling dimension $\Theta(\log n)$, such that with probability at least $2/3,$ $\widetilde{M} = o(1) \cdot M$.
\end{theorem}

\begin{theorem}[MST] \label{thm:lb_mst_pullback}
    Let $d = o(\log n).$ There exists $X$ with doubling dimension $\Theta(\log n)$, such that with probability at least $2/3,$ $\widetilde{\mathcal{M}},$ when pulled back to $X$, will have cost $\omega(1) \cdot M$.
\end{theorem}

Next, we show that (local) optimality is required for Theorems \ref{cor:main} and \ref{thm:MainMST}, and cannot be replaced with \emph{approximate} optimality. In other words, random projections to $o(\log n)$ dimensions do not necessarily preserve the set of approximate solutions for either facility location or MST, even for point sets of low doubling dimension. Namely, we show 
the following two lemmas:
%Lemmas \ref{lem:opt_nec} and \ref{lem:opt_nec_mst}, which we now state.

\begin{lemma}[FL] \label{lem:opt_nec}
    Let $d = o(\log n)$. There exists $X$ with constant doubling dimension, such that with at least $2/3$ probability, there exists a $(1+O(m^{-1/2d})) = (1+o(1))$-approximate solution $\F'$ for $GX$ whose total cost when pulled back to $X$ is at least $\Omega(m^{1/2d}) \cdot F = \omega(1) \cdot F$.
\end{lemma}

\begin{lemma}[MST] \label{lem:opt_nec_mst}
    Let $d = o(\log n)$ but $d = \omega(\log \log n)$. There exists $X$ with constant doubling dimension, such that with at least $2/3$ probability, there exists a $(1+o(1))$-approximate MST $\mathcal{M}'$ for $GX$ whose total cost whose total cost when pulled back to $X$ is at least $\omega(1) \cdot M$.
\end{lemma}

Finally, we show that the guarantees of facility location are in fact not maintained for $k$-means and $k$-medians clustering. In other words, the bound of $O(\log k)$ by \cite{ilyapaper} is optimal even for sets of doubling dimension $O(1)$.

\begin{theorem}[$k$-means/$k$-medians] \label{thm:lb_kmeans}
    Let $k < n$ and $d = o(\log k)$. Then, there exists $X$ with constant doubling dimension, such that with probability at least $2/3$, the $k$-means (resp., medians) cost of $GX$ is $o(1)$ times the $k$-means (resp., medians) cost of $X$. Moreover, the optimal choice of $k$ centers in $GX$, when pulled back to $X$, will be an $\omega(1)$-approximate solution in the original dimension $\mathbb{R}^m$.
\end{theorem}

At a first glance, Theorem \ref{thm:lb_kmeans} may appear to contradict our upper bounds for facility location. However, in our counterexamples for $k$-means and $k$-medians, the cost (both initially and after projection) is substantially smaller than $k$. Facility location adds a cost of $k$ for the $k$ facilities that are created, and since these facilities now make up the bulk of the cost, the facility location cost is still approximately preserved under random projection.

\section{Experiments} \label{sec:experiments}
We use the following datasets in our experiments for Subsections \ref{subsec:exp_FL} and \ref{subsec:exp_MST}.
%\vspace{-1mm}
\begin{itemize}
    \item \textbf{Faces Dataset}: This dataset is used in the influential ISOMAP paper and consists of $698$ images of faces in dimension $4096$ \cite{isomap}. From \cite{faces}, we can estimate that the doubling dimension of this dataset is a small constant.  
    \item \textbf{MNIST `2' Dataset}: $1000$ randomly chosen images from the MNIST dataset (dimension 784) restricted to the digit $2$. We picked $2$ since it is considered in the original ISOMAP paper \cite{isomap}.
\end{itemize}
%\vspace{-1mm}
All of our experimental results are averaged over $20$ independent trials and the projection dimension $d$ ranges from $5$ to $20$ inclusive.All of our experiments were done on a CPU with i5 2.7 GHz dual core and 8 GB RAM. 
%\todo{Link to source code} \todo{link to datasets}

%\ifthenelse{\value{num}=1}{
\begin{figure*}[!ht]
    \centering
      \subcaptionbox{Faces Dataset\label{fig:FL_a} }[0.65\columnwidth][c]{%
    %   \hspace{-3.175cm}
     \includegraphics[width=2.2in]{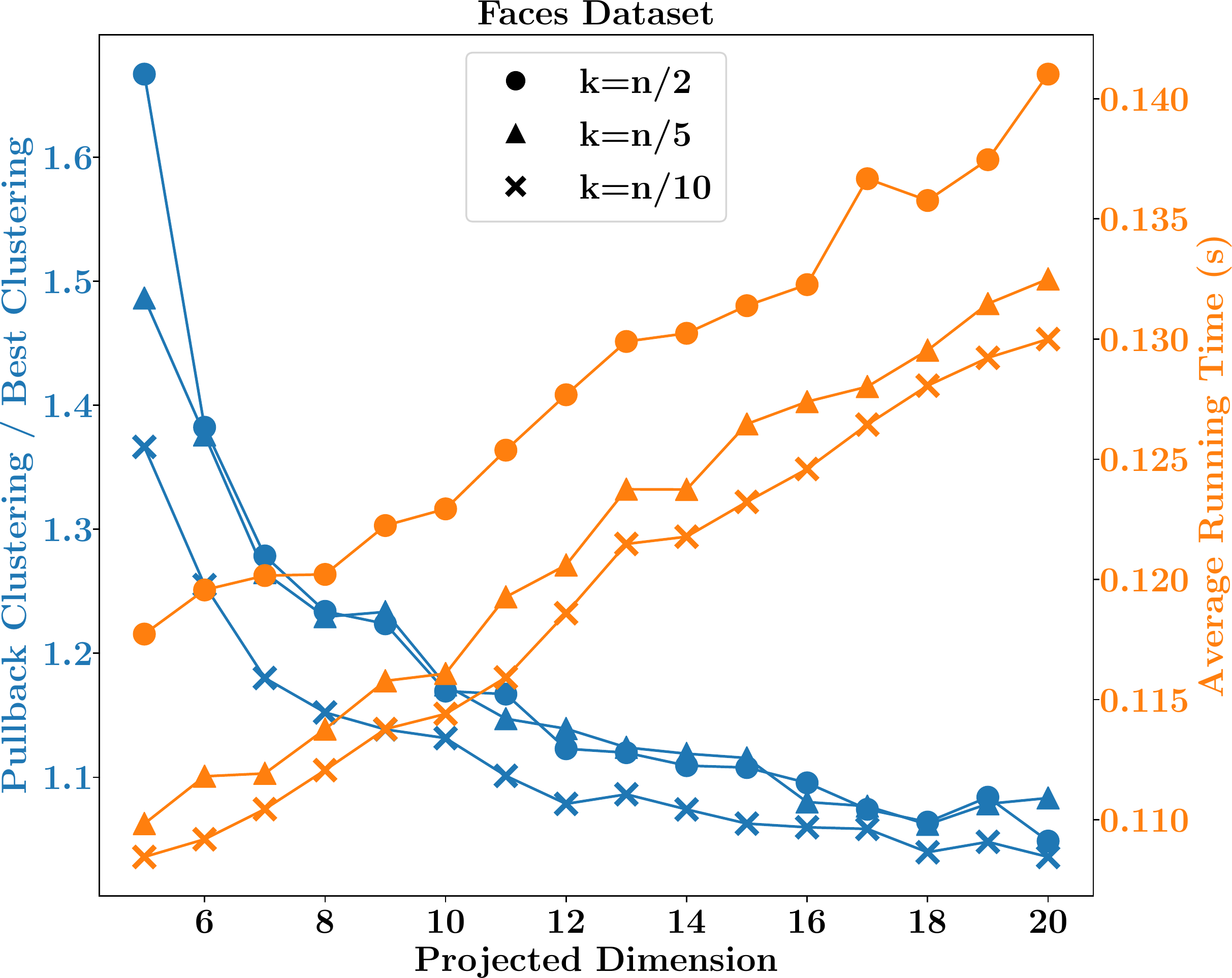}}
     \hspace{35mm}
      \subcaptionbox{MNIST `2' Dataset\label{fig:FL_b}}[0.65\columnwidth][c]{%
        \includegraphics[width=2.2in]{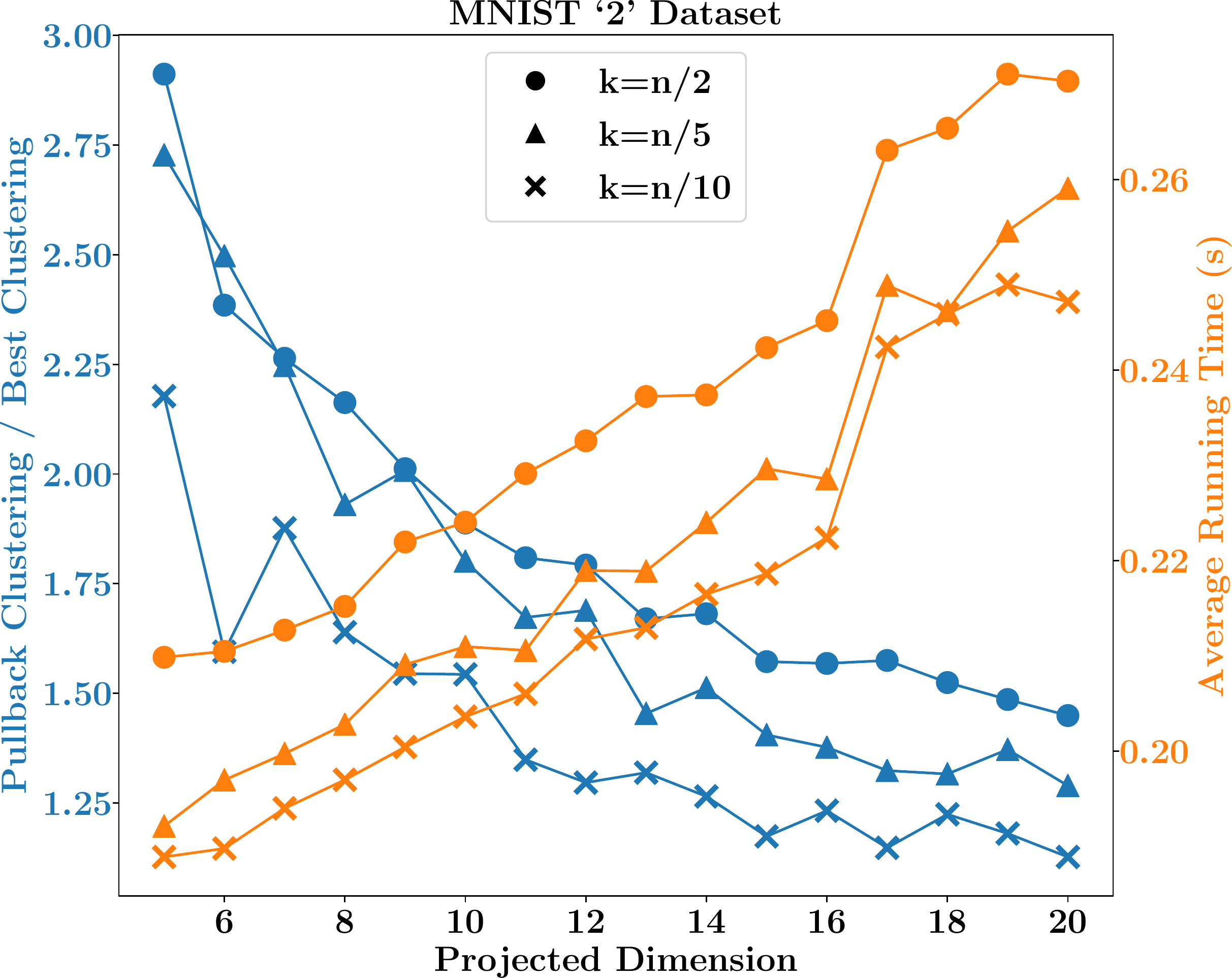}}
      \caption{Facility Location Experiments. (a) \textbf{Blue:} Ratio of solution costs with/without dimensionality reduction, as a function of $d$. \textbf{Orange: }Running time (in secs) as a function of $d$. (b) Same plot as (a) but for MNIST `2' dataset.}
      \label{fig:FL_experiments} 
    \end{figure*}
 \begin{figure*}[!ht]
      \centering
      \subcaptionbox{Faces Dataset\label{fig:MST_a} }[0.5\columnwidth][c]{%
      \hspace*{-1cm}
     \includegraphics[width=2.2in]{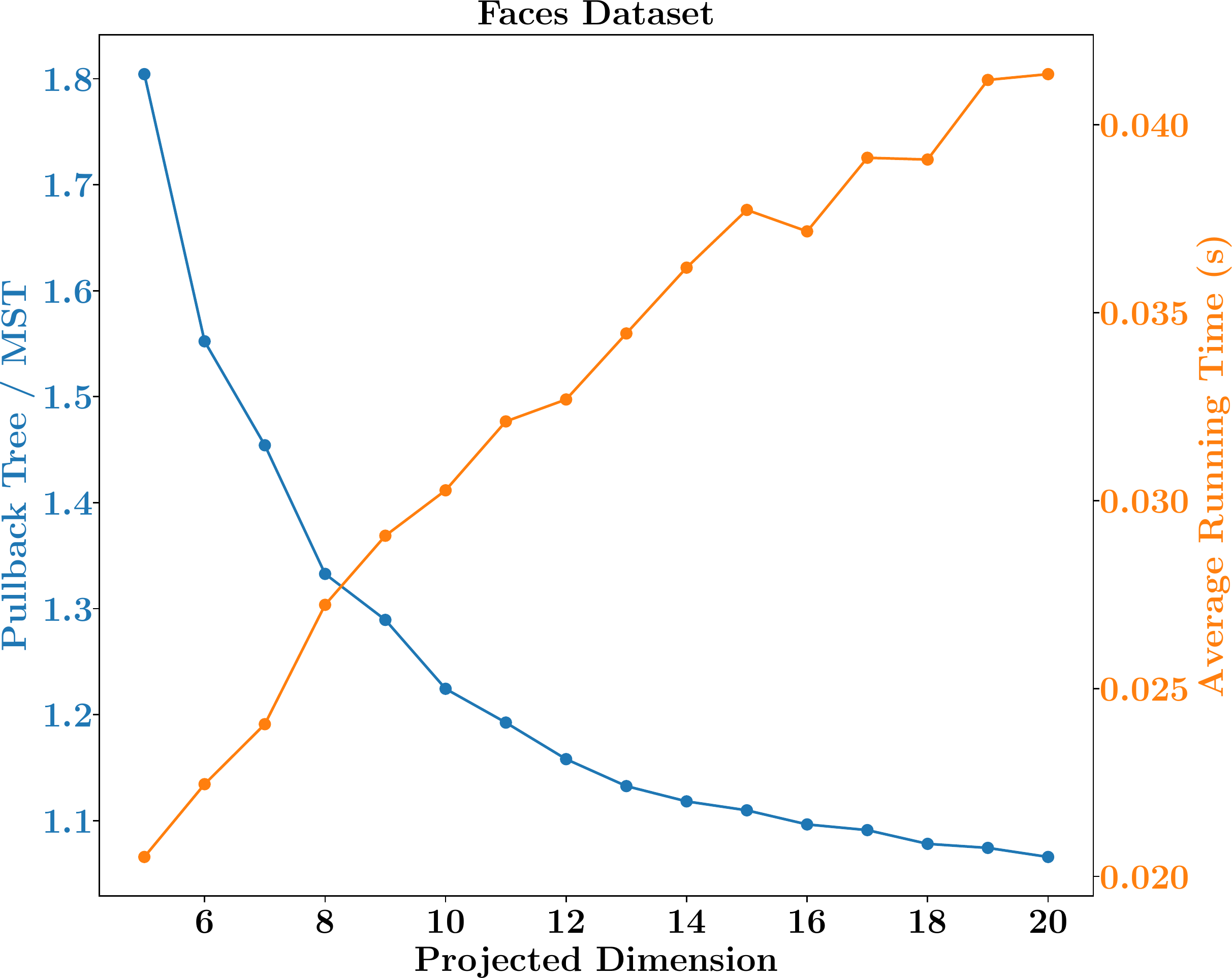}}\hspace{17mm}
      \subcaptionbox{MNIST `2' Dataset\label{fig:MST_b}}[0.5\columnwidth][c]{%
      \hspace*{-1cm}
        \includegraphics[width=2.2in]{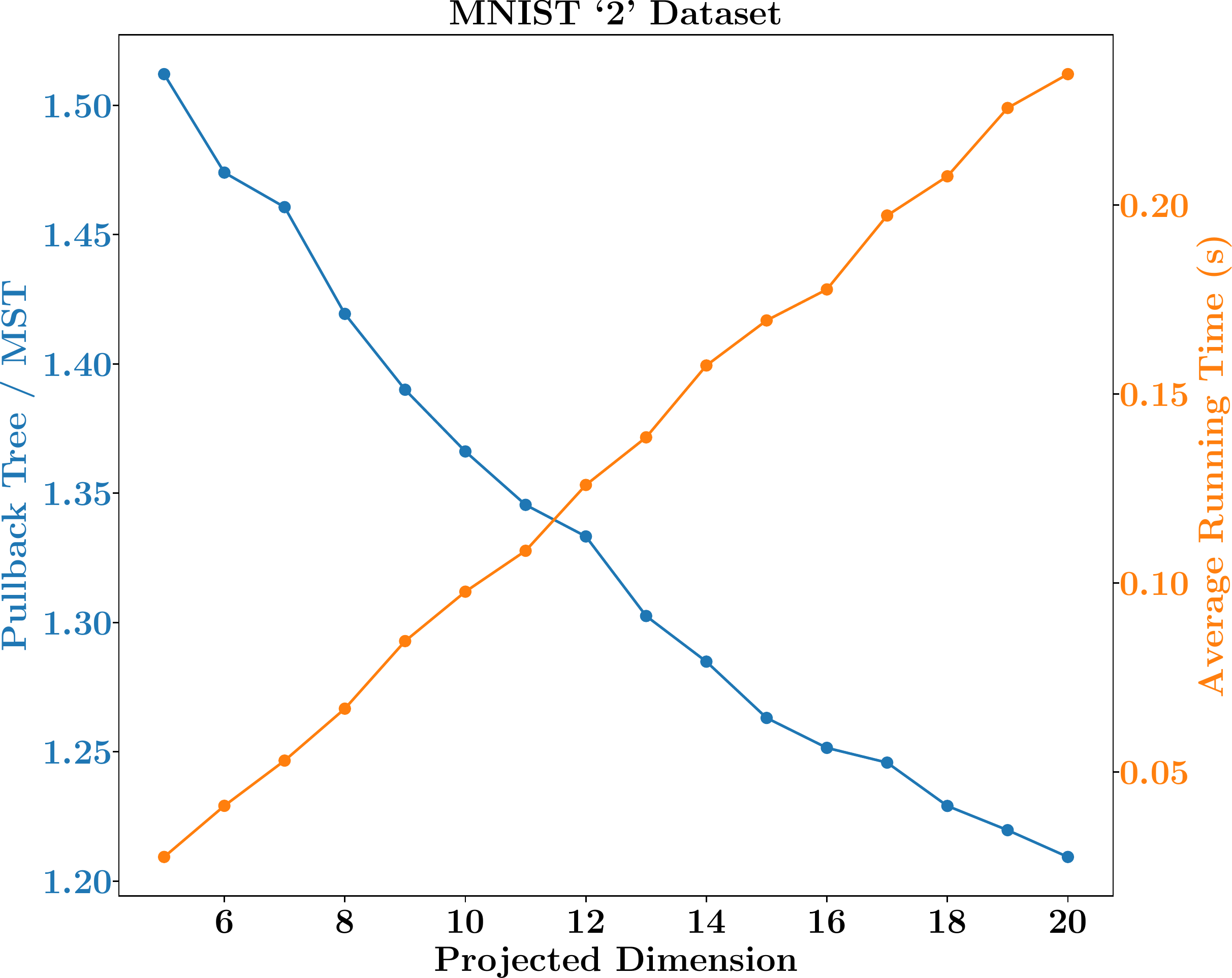}}\hspace{17mm}
      \subcaptionbox{Low vs High Doubling Dimension Comparison\label{fig:large_vs_small}}[0.5\columnwidth][c]{%
      \hspace*{-1cm}
        \includegraphics[width=2.2in]{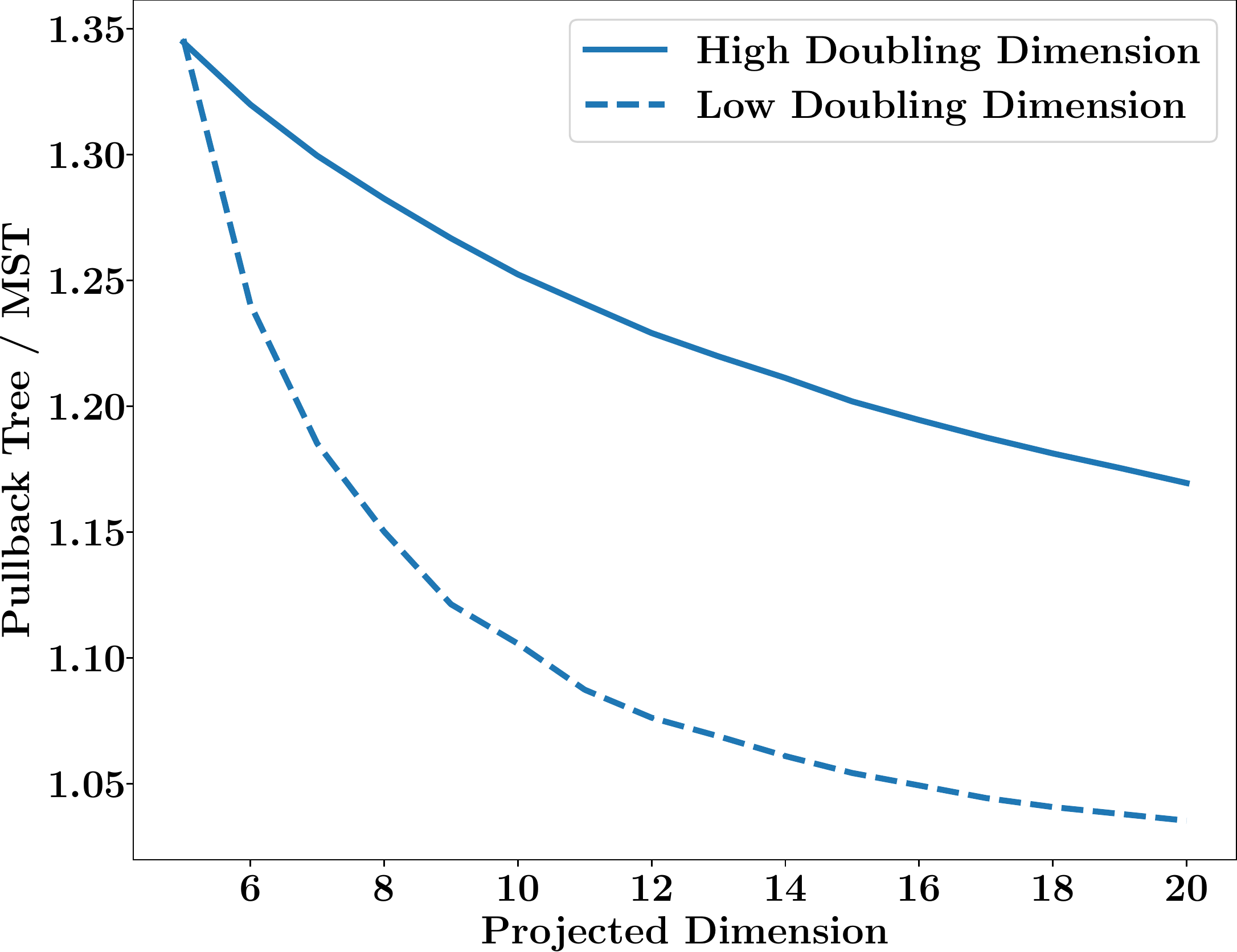}}
      \caption{Minimum Spanning Tree Experiments. (a) \textbf{Blue:} Ratio of solution costs with/without dimensionality reduction, as a function of $d$. \textbf{Orange:} Running time (in secs) as a function of $d$. (b) Same plots as (a) but for MNIST `2' dataset. (c) Dataset $1$ (low doubling dimension) can be projected into a much smaller dimension than Dataset $2$ for MST computation.}
      \label{fig:MST_experiments} 
    \end{figure*}
%}{}

\subsection{Facility Location: Cost versus Accuracy Analysis} \label{subsec:exp_FL}
In this section we compare the accuracy of the MP algorithm with/without dimensionality reduction for various number of centers opened. 
%\vspace{-1mm}
\paragraph{Experimental Setup }
We project our datasets and compute a facility location clustering with the opening costs scaled so that $n/2, n/5,$ and $n/10$ facilities are opened respectively. We then take this solution and evaluate its cost in the original dimension. We also perform a clustering in the original dimension with the same prescribed number of facilities opened and plot the ratio of the cost of the solution found in the lower dimension (but evaluated in the larger dimension) to the solution found in the larger dimension. We also plot the time taken for the clustering algorithm in the projected dimension. We use the MP algorithm to perform our clustering due to the intractability of finding the exact optimum and also because the MP algorithm is fast and quite practical to use.
%\vspace{-1mm}
\paragraph{Results}
%\ifthenelse{\value{num}=1}{
Our results are plotted in Figures \ref{fig:FL_a}-\ref{fig:FL_b}.
%}
%{Our results are plotted in Figure \ref{facility_fig_cost}.}
Our experiments empirically demonstrate that the dimensionality reduction step does not significantly decrease the accuracy of the solution. Furthermore, we get a substantial reduction in the runtime since the average runtime was at least $20$ seconds for Faces and around  $6.5$ seconds for MNIST `2' in the original dimension for all the values of $k$ tested, which is \textbf{1-2} orders of magnitude higher than the runtime when random projections are used. Note that the runtime includes the time taken to perform the random projection. Overall, our experiments demonstrate that the method of performing dimensionality reduction to perform facility location clustering is well-founded.
%\vspace{-1mm}

\subsection{MST: Cost versus Accuracy Analysis} \label{subsec:exp_MST}

We empirically show the benefits of using dimensionality reduction for minimum spanning tree computation.
%\vspace{-1mm}
\paragraph{Experimental Setup}
We project our datasets and compute a MST. We then take the tree found in the lower dimension and compare its cost in the higher dimension against the actual MST. Our MST algorithm is a variant of the Boruvka algorithm from \cite{MST_mlpack} that is suitable for point sets in large dimensions and is implemented in the popular `mlpack' machine learning library \cite{mlpack2018}.
%\vspace{-1mm}
\paragraph{Results}
%\ifthenelse{\value{num}=1}{
Our results are plotted in Figures \ref{fig:MST_a}-\ref{fig:MST_b}. In the blue plots of these figures, the ratio of the cost of the tree found in the projected dimension, but evaluated in the original dimension, to the cost of the actual MST is shown. We see that indeed as projection dimension increases, the ratio approaches $1$. However even for very low values of $d$, such as $10$, the tree found in the projected space serves as a good approximate for the actual MST. Conversely, we see that as $d$ increases, the cost of computing the MST also increases as shown in the orange plots of the Figures \ref{fig:MST_a} and \ref{fig:MST_b}. Note that the time taken to perform the projection is also included. The time taken to compute the MST in the original dimension was approximately $3.2$ seconds for the Faces dataset and $7.1$ seconds for the MNIST `2' dataset. Therefore, projection to dimension $d=20$ gives us approximately \textbf{80x} improvement in speed for the Faces dataset and \textbf{30x} improvement in speed for the MNIST `2' dataset while having a low cost distortion.
%}

\subsection{Large versus Small Doubling Dimension } \label{subsec:exp_MST_doubling}
In this section we present two datasets in $\mathbb{R}^n$ where one dataset has doubling dimension $O(1)$ and the other has doubling dimension at least $\Omega(\log n)$ which is asymptotically the largest doubling dimension of any set of size $n$. We empirically show that the second dataset requires larger projection dimension than the first to guarantee that the MST found in the projected space induces a good solution in the original space. Our two datasets are the following. Let $e_i$ denote the standard basis vectors in $\mathbb{R}^n$. We first draw $n$ standard Gaussians $g_1, \cdots, g_n \in \mathbb{R}$. Our datasets are:

\noindent \textbf{Dataset 1:} $\{ g_1 \cdot e_1, g_1 \cdot e_1 + g_2 \cdot e_2, \ldots, g_1 \cdot e_1 + \cdots + g_n \cdot e_n\}$. \\
\noindent \textbf{Dataset 2:} $\{g_1 \cdot e_1, g_1 \cdot e_2, \ldots, g_n \cdot e_n \}$. 

Note that we use the same $g_i$'s for both datasets. The above datasets appear to be similar, but it can be shown that their respective doubling dimensions are $O(1)$ and $\Omega(\log n)$.
%\vspace{-1mm}
\paragraph{Experimental Setup} We let $n = 1000$ and construct the two datasets. We project our datasets and find the MST for each dataset in the projected space. Then we evaluate the cost of this tree in the larger dimension and compare this cost to the cost of the actual MST for each dataset. 
%\vspace{-1mm}
\paragraph{Results} Figure \ref{fig:large_vs_small} demonstrates that we can find a high quality approximation of the MST by finding the MST in a much smaller dimension for Dataset $1$ compared to Dataset $2$. For example, Dataset $1$ required only $d = 10$ dimensions to approximate the true MST within $10\%$ relative error while Dataset $2$ needed $d = 38$ to get within $10\%$ relative error of the true MST.

\section*{Acknowledgments} This  research  was  supported  in  part  by  the  NSF  TRIPODS  program  (awards  CCF-1740751  and DMS-2022448);  NSF  award  CCF-2006798; MIT-IBM Watson collaboration;  Simons Investigator Award; and NSF Graduate Research Fellowship Program.

\appendix

\section{Omitted Preliminaries} \label{sec:app_prelim}

In this section, we state all of the preliminary results needed relating to random projections that were omitted in Section \ref{sec:prelim}. In all of the following results, we treat $G$ as a random projection from $\mathbb{R}^m$ to $\mathbb{R}^d$.

First, if $x \in S^{m-1}$ then the following statements hold about the distribution of $\|Gx\|$ \cite{indyknaor}:
\begin{align}
    \Pr( | \|Gx\| -1 | \ge t) &\le \exp(-dt^2/8), \label{eq:Gxa} \\
    \Pr( \|Gx\| \le 1/t) &\le \left( \frac{3}t \right)^d. \label{eq:Gxb}
\end{align}
The following serves as a converse to Equation \eqref{eq:Gxb}.
\begin{proposition} \label{prop:ChiLB}
If $x \in S^{m-1}$ and $t \ge 1$ then the following is true about the distribution of $\|Gx\|$:
\begin{equation}\label{eq:Gx2}
    \Pr( \|Gx\| \le 1/t) \ge \left( \frac{1}{et} \right)^d.
\end{equation}
\end{proposition}

\begin{proof}
    Since $x \in S^{m-1}$, $d \cdot \|G x\|^2$ is a chi-squared random variable with $d$ degrees of freedom so it has density
\[\frac{1}{2^{d/2} \Gamma(d/2)} x^{d/2-1} \exp(-x/2).\]
    Thus, for all $t \le 1,$ the probability that $\|Ga\|^2$ is less than $1/t^2$ is at least
\begin{align*}
    &\hspace{0.5cm} \frac{1}{2^{d/2} \Gamma(d/2)} \int_0^{d/t^2} x^{d/2-1} \exp(-x/2) dx \\
    &\ge \frac{\exp(-d/2t^2)}{2^{d/2} \Gamma(d/2)} \int_0^{d/t^2} x^{d/2-1} dx \\
    &\ge \frac{\exp(-d)}{2^{d/2}\Gamma(d/2) (d/2)} \cdot \left(\frac{d}{t^2}\right)^{d/2} \\
    &\ge \frac{\exp(-d)}{2^{d/2} \cdot (d/2)^{d/2}} \cdot \left(\frac{d}{t^2}\right)^{d/2} \\
    &= \left(\frac{1}{e \cdot t}\right)^d,
\end{align*}
    where we used the well-known fact that $\Gamma(x) \cdot x \le x^x$ for all $x \ge 1$.
\end{proof}

We will need the following lemma to prove some of our lower bound results from Section \ref{sec:lower}.

\begin{lemma} \label{Nearby}
    Let $C \ge 1$ and fix some point $v$ of norm at most $C$ in $\mathbb{R}^d$. Then, if $x \sim \frac{1}{\sqrt{d}} \cdot \mathcal{N}(0, I_d)$ is a $d$-dimensional scaled multivariate Normal, then $\Pr(\|x-v\| \le \frac{1}{C}) \ge n^{-1/10},$ if $d \le \log n/(10 C^2)$ and $n$ is sufficiently large.
\end{lemma}

\begin{proof}
    By the rotational symmetry of the multivariate normal, assume $v = (r, 0, \dots, 0),$ where $0 \le r \le C.$ Then, if $x = (x_1, y)$ for $x_1 \in \mathbb{R}, y \in \mathbb{R}^{d-1}$, then if $r-\frac{1}{2C} \le x_1 \le r$ and $\|y\| \le \frac{1}{2C},$ then we indeed have $\|x-v\| \le \frac{1}{C}.$ Since $\sqrt{d} x_1 \sim \mathcal{N}(0, 1)$ and $r \le C$, the probability that $r - \frac{1}{2C} \le x_1 \le r$ equals the probability that $\mathcal{N}(0, 1) \in [(r-1/2C) \sqrt{d}, r \sqrt{d}],$ which is at least $\frac{\sqrt{d}}{2C} \cdot \frac{1}{\sqrt{2 \pi}} \cdot e^{-C^2 d/2}.$ Moreover, the probability that $\|y\| \le \frac{1}{2C}$ is at least $\left(\frac{1}{2eC}\right)^d$ by Proposition \ref{prop:ChiLB}. Therefore,
\begin{align*}
\Pr\left(\|x-v\| \le \frac{1}{C}\right) &\ge \frac{\sqrt{d}}{2C} \cdot \frac{1}{\sqrt{2 \pi}} \cdot e^{-C^2 d/2} \cdot \left(\frac{1}{2 e C}\right)^d \\
&\ge n^{-1/10},
\end{align*}
    where the last inequality is true because $d \le \log n/(10 C^2)$ and that $n$ is sufficiently large.
\end{proof}

We recall Lemma \ref{lem:indyk}, due to Indyk and Naor \yrcite{indyknaor}.

\begin{lemma}[Lemma \ref{lem:indyk}] Let $X \subseteq B(0,1)$ be a subset of the $m$-dimensional Euclidean unit ball. Then there exist universal constants $c, C > 0$ such that for $d \ge C \cdot d_X + 1$ and $t > 2$, $\Pr(\exists x \in X, \, \|Gx\| \ge t) \le \exp(-cdt^2)$.
\end{lemma}

Indyk and Naor also prove the following result about the distance to the nearest neighbor after a random projection.

\begin{theorem}[Theorem $4.1$ in \cite{indyknaor}]\label{thm:indyk}
Let $G$ be a random projection from $\mathbb{R}^m$ to $\mathbb{R}^d$ for $d = O(d_X  \cdot \log(1/\epsilon)/\epsilon^2 \log(1/\delta))$. Then for every $x \in X$, with probability at least $1-\delta$, the following statements hold:
\begin{enumerate}
    \item $D(Gx, G(X \setminus \{x\})) \le (1+\epsilon)D(x, X \setminus \{x \})$
    \item Every $y \in X$ with $\|x-y\| > (1+2\epsilon)D(x, X\setminus\{x\})$ satisfies $\|Gx-Gy\| > (1+\epsilon)D(x, X\setminus\{x\})$
where $D(x, X) = \min_{y \in X} \|x-y\|.$
\end{enumerate}
\end{theorem}

\section{The Mettu-Plaxton (MP) Algorithm and Local Optimality} \label{sec:app_local}

First, we give the pseudocode for the Mettu-Plaxton (MP) algorithm for facility location, described in Section \ref{sec:local}.

\begin{minipage}{\linewidth}
\RestyleAlgo{boxruled}
\removelatexerror
\begin{algorithm}[H]
	\SetKwInOut{Input}{Input}
	\SetKwInOut{Output}{Output}
\Input{Dataset $X = \{p_1, \cdots, p_n\} \subseteq \mathbb{R}^d$}
	\Output{Set $\mathcal{F}$ of facilities}
	\DontPrintSemicolon
	$\mathcal{F} \gets \emptyset$ \;
	\For{$i = 1$ to $n$}{
    Compute $r_i$ satisfying:
    	$ \sum_{q \in B(p_i, r_i)} (r_i - \|p_i-q\|) = 1$
	\;
    }
    Sort such that $r_1 \le \ldots \le r_n$\;
    \For{$i=1$ to $n$}{
    \If{$B(p_i, 2r_i) \cap \mathcal{F} = \emptyset$}{
    $\mathcal{F} \gets \mathcal{F} \cup \{p_i\}$}
    }
    Output $\mathcal{F}$
	\caption{$\textsc{MP Algorithm}$}
	\label{alg:MPalg}
\end{algorithm}
\end{minipage}

\bigskip

Next, we prove Lemma \ref{lem:localopt}, which roughly stated that a globally optimal solution for facility location is always locally optimal.

\begin{proof}[Proof of Lemma \ref{lem:localopt}]
Consider an arbitrary point $p \in X$. We first establish a lower bound on the number of points in $B(p, r_p) \cap X$. Note that by definition of $r_p$, we have
$$|B(p, r_p) \cap X|r_p \ge \sum_{q \in B(p,r_p)} (r_p-\|p-q\|) = 1 $$
so it follows that 
\begin{equation}\label{eq:numberlb}
   |B(p, r_p) \cap X| \ge 1/r_p.
\end{equation}
Now suppose that $B(p, 3r_p) \cap  \mathcal{F} = \emptyset$ and let $m$ be the number of points in $|B(p, r_p) \cap X|$ \emph{excluding} $p$. The total connection cost of all these points to their nearest facility must be at least $2mr_p$. Accounting for point $p$, the total connection costs of points in $B(p, r_p) \cap X$ is at least $2mr_p + 3r_p$. Now if we open a new facility at point $p$, then the connection costs of these points is at most $mr_p$ but we also incur an additional cost for opening a facility at $p$. Therefore, the total cost of the solution decreases by at least
$$(2mr_p + 3r_p) - (1 + mr_p) = (m+3)r_p -1. $$
Now from Eq.\@ \eqref{eq:numberlb}, we have that $(m+3)r_p > 1$, which means that the total cost decreases if we open a new facility at $p$.
\end{proof}

\section{Dimension Reduction for Facility Location: Omitted Proofs} \label{sec:app_fl}
\subsection{Approximating the Optimal Facility Location Cost}\label{subsec:approx_appendix}
In this subsection, we prove Theorem \ref{thm:constant}. As stated in Subsection \ref{subsec:approx}, our proof involves computing an upper bound and a lower bound for $\mathbb{E}[\tilde{r}_p]$. We first proceed with an upper bound in Lemma \ref{lem:rupper}.

\begin{lemma}\label{lem:rupper}
Let $X \subseteq \mathbb{R}^m$ and let $p \in X$. Let $G$ be a random projection from $\mathbb{R}^m$ to $\mathbb{R}^d$ for $d = O(\log \lambda_X \cdot \log(1/\epsilon)/\epsilon^2)$. Let $r_p$ and $\widetilde{r}_p$ be the radius of $p$ and $Gp$ in $\mathbb{R}^m$ and $\mathbb{R}^d$ respectively, computed according to Eq.\@ \eqref{eq:rdef}. Then $$\E[\tilde{r}_p] \le (2+O(\epsilon))r_p.$$
\end{lemma}
\begin{proof}
Let $\delta > 0$ be fixed and let $\mathcal{E}_k$ be the event that
$$\max_{x \in B(p, r_p) \cap X} \|G(x-p)\| \in [(k-1)(1+\delta)r_p, k(1+\delta)r_p)$$
Note that $\mathcal{E}_k$ implies that there exists an $x \in B(p, r_p) \cap X$ such that $\|G(x-p)\| \ge (k-1)(1+\delta)r_p$, so by Lemma \ref{lem:indyk} we have
\begin{align}\label{eq:rpprob}
  \Pr(\mathcal{E}_k) &\le \Pr(\exists x \in B(p, r_p) \cap X, \, \|G(x-p)\| \ge k(1+\delta)r_p) \nonumber \\
  &\le \exp(-c(k-1)^2(1+\delta)^2d)  
\end{align}
for some constant $c$. We now show that conditioned on $\mathcal{E}_k$, we have $\tilde{r}_p \le (k+1)r_p(1+\delta)$. This is because conditioning on $\mathcal{E}_k$ gives us
\begin{multline*}
\sum_{Gq \in B(Gp, (k+1)(1+\delta)r_p)} ((k+1)r_p(1+\delta)-\|Gp-Gq\|) \\
\ge \sum_{q \in B(p, r_p)} (1+\delta)r_p,
\end{multline*}
where $Gq \in B(Gp, (k+1)(1+\delta)r_p)$ is interpreted as summing over the points in the set $$GX \cap B(Gp, (k+1)(1+\delta)r_p).$$
Furthermore,
\begin{align*}
\sum_{q \in B(p, r_p)} r_p(1+\delta) &\ge \sum_{q \in B(p, r_p)} r_p \\
&\ge \sum_{q \in B(p, r_p)} (r_p - \|p-q\|) = 1.
\end{align*}
Therefore, by the observation that the function  $f(r) = \sum_{q \in B(p, r)} (r - \|p-q\|)$ is increasing in $r$, it follows that $(k+1)(1+\delta)r_p \ge \tilde{r}_p$. Therefore, we have
\begin{equation}\label{eq:rpcond}
    \E[\tilde{r}_p \mid \mathcal{E}_k] \le (k+1)(1+\delta)r_p.
\end{equation}
Now using Eq.\@ \eqref{eq:rpprob}
\begin{align*}
    \E[\tilde{r}_p] &= \sum_{k=1}^{\infty} \E[\tilde{r}_p \mid \mathcal{E}_k]\Pr(\mathcal{E}_k) \\
    &\le (2+\delta)r_p \\
    &\hspace{1cm}+ r_p\sum_{k=2}^{\infty}(k+1) \exp(-c(k-1)^2(1+\delta)^2d) \\
    &\le (2+\delta)r_p \\
    &\hspace{1cm}+ (1+\delta)r_p\int_0^{\infty}(x+2)\exp(-C'x^2) \ dx
\end{align*}
where $C' = c(1+\delta)^2d$. We can explicitly evaluate that
$$ \int_0^{\infty} (x+2)\exp(-C'x^2) \ dx  = \sqrt{\frac{\pi}{C'}} + \frac{1}{2C'}.$$
Noting that $d = \Omega(1/\epsilon^2)$, we have that 
$$ \E[\tilde{r}_p] \le (2+O(\epsilon))r_p $$
by picking $\delta = O(\epsilon)$.
\end{proof}

We now show the corresponding lower bound.

\begin{lemma}\label{lem:rlower}
Let $X \subseteq \mathbb{R}^m$ and let $p \in X$. Let $G$ be a random projection from $\mathbb{R}^m$ to $\mathbb{R}^d$ for $d = O(\log \lambda_X \cdot  \log(1/\epsilon)/\epsilon^2)$. Let $r_p$ and $\widetilde{r}_p$ be the radius of $p$ and $Gp$ in $\mathbb{R}^m$ and $\mathbb{R}^d$ respectively, computed according to Eq.\@ \eqref{eq:rdef}. Then 
$$ \E[\tilde{r}_p] \ge \frac{(1-\epsilon)r_p}4.$$
\end{lemma}
\begin{proof}
Let $k$ be the size of the set $|B(p, r_p/2) \cap X|$. By definition of $r_p$, the following inequality holds:
\begin{align}\label{eqri}
    1 &= \sum_{q \in B(p, r_p)}(r_p-\|p-q\|) \nonumber \\
    &\ge \sum_{q \in B(p, r_p/2)}(r_p-\|p-q\|) \ge \frac{k r_p}2.
\end{align}
Now let $\mathcal{E}$ be the event that the ball $B(Gp, (1-\epsilon)r_p/2)$ contains at most $k$ points. 
%Now by a similar reasoning as in Theorem \ref{thm:indyk}, we have $\Pr(\mathcal{E}) \ge 1/2$. This is because we can invoke Theorem \ref{thm:indyk} in a black-box fashion.
By invoking Theorem \ref{thm:indyk}, we will show that $\Pr(\mathcal{E}) \ge 1/2$.
Consider the set $X$ without the $k-1$ points in $B(p,r_p/2)  - \{p\}$, and with an extra point $q$ at distance $(1-\epsilon)r_p/2$ from $p$. The added point $q$ becomes a nearest neighbor of $p \in X$. By Theorem \ref{thm:indyk} part (2) applied to an appropriately chosen $\epsilon'= O(\epsilon)$,  with probability at least $1/2$, no point outside of $B(p,r_p/2)$ is mapped within $(1-\epsilon)r_p/2$ of $p$. After removing $q$, only the originally removed  $k-1$ points (and $p$) can lie in $B(p,r_p(1-\epsilon)/2)$.

Conditioning on $\mathcal{E}$, it follows that 
\begin{multline}\label{eqri2}
    \sum_{Gq \in B(Gp, (1-\epsilon)r_p/2)}\left(\frac{(1-\epsilon) r_p}2-\|Gp-Gq\| \right) \\
    \le \frac{(1-\epsilon)kr_p}{2} .
\end{multline}
Combining Eq.\@ \eqref{eqri2} with \eqref{eqri}, we have that 
$$ \sum_{Gq \in B(Gp, (1-\epsilon)r_p/2)}\left(\frac{(1-\epsilon) r_p}2-\|p-q\| \right) \le 1-\epsilon < 1. $$ 
Therefore, conditional on $\mathcal{E}$, it follows that $\widetilde{r}_p \ge (1-\epsilon)r_p/2$. Hence,
\begin{equation*}
    \mathbb{E}[\tilde{r}_p] \ge \frac{ \mathbb{E}[\tilde{r}_p \mid \mathcal{E}]}2 \ge \frac{(1-\epsilon)r_p}{4}.
    \qedhere
\end{equation*}
\end{proof}
Combining Lemma \ref{lem:rupper} and \ref{lem:rlower} gives us the complete proof of Theorem \ref{thm:constant}. 
\begin{proof}[Proof of Theorem \ref{eq:rdef}]
The theorem follows from combining the result given in Lemma \ref{lem:appx}, that the sum of the radii $r_p$ is a constant factor approximation to the global optimal solution, and Lemmas \ref{lem:rupper} and \ref{lem:rlower}, which state that $\mathbb{E}[\tilde{r}_p]$ is a constant factor approximation to $r_p$.
\end{proof}

\subsection{Obtaining a Solution to Facility Location in Larger Dimension}
Recall that the main technical challenge is to show that if a facility is within distance $O(\tilde{r}_p)$ of a fixed point $p$ in $\mathbb{R}^d$ (note that $\tilde{r}_p$ is calculated according to Eq.\@ \eqref{eq:rdef}, in $\mathbb{R}^d$), then the facility must also be within distance $O(r_p)$ in $\mathbb{R}^m$, the larger dimension. We prove this claim formally in Theorem \ref{thm:main}. 

Before presenting Theorem \ref{thm:main}, we need the following technical result later on for our probability calculations.

\begin{lemma}\label{lem:tech}
Denote $\textup{erf}(x)$ to be the error function defined as
$$\textup{erf}(x) = \frac{2}{\sqrt{\pi}} \int_0^x \exp(-t^2) \ dt.$$ Then,
$$1-\textup{erf}(x) \le \exp(-x^2) $$
for all $x \ge 1$.
\end{lemma}
\begin{proof}
Note that $f(x) \exp(-t^2)/\sqrt{\pi}$ is a valid probability density function over $\mathbb{R}$ so that
$$  \text{erf}(x) = 1- \Pr(|Z| \ge x)$$
where $Z$ is distributed according to the density $f$. Now
\begin{align*}
    \Pr(Z \ge x) &= \frac{1}{\sqrt{\pi}} \int_x^{\infty} \exp(-t^2) \ dt\\
    &\le \frac{1}{\sqrt{\pi}} \int_x^{\infty} \frac{t}x \exp(-t^2) \ dt \\
    &= \frac{\exp(-x^2)}{2x\sqrt{\pi}}
\end{align*} 
where the inequality follows from the fact that $t \ge x$. By symmetry, we have
$$ \text{erf}(x) + \exp(-x^2) -1 \ge \exp(-x^2)\left(1 - \frac{1}{x\sqrt{\pi}} \right) \ge 0 $$
for $x \ge 1$.
\end{proof}

The proof of Theorem \ref{thm:main} relies on the careful balancing of the following two events. First, we control the value of the radius $\tilde{r}_p$ and show that $\tilde{r}_p \approx r_p$. In particular, we show that the probability of $\tilde{r}_p \ge k r_p$ for any constant $k$ is exponentially decreasing in $k$. The argument for this part follows similarly to the argument in Lemma \ref{lem:rupper}.

Next, we need to bound the probability that a `far' point comes `close' to $p$ after the dimensionality reduction. While Theorem \ref{thm:indyk} roughly states that `far' points do not come too `close', we need a more detailed result to quantify how close far points can come after the dimension reduction. 

To study this in a more refined manner, we bucket the points in $X \setminus \{p\}$ according to their distance from $p$. The distance spacing between buckets will be a linear scale. We show that points in $X\setminus \{p\}$ that are in `level' $i$ do not shrink to a `level' smaller than $O(\sqrt{i})$. Note that we need to control this even across all levels. To do this requires a chaining type argument which crucially depends on the doubling dimension of $X$. Finally, a careful combination of probabilities gives us our result.

\begin{theorem}\label{thm:main}
Let $X\subseteq \mathbb{R}^m$ and let $G$ be a random projection from $\mathbb{R}^m$ to $\mathbb{R}^d$ for $d = O(\log \lambda_X \cdot  \log(1/\epsilon)/\epsilon^2)$. Fix $p \in X$ and let $x \in X$ be the point that maximizes $\|p-x\|$ subject to the condition $Gx \in B(Gp, C \tilde{r}_p)$ where $C$ is a fixed constant. Then
$$ \E \|p-x\| \le 2C(1+O(\epsilon)) r_p. $$
\end{theorem}
\begin{proof} For simplicity, let $r = r_p, \tilde{r} = \tilde{r}_p,$ and define $t_{-1} = 0, t_0 = 1,$ and $$t_i = 1+2\epsilon +  \frac{\epsilon (i-1)}4 $$  for all $i \ge 1$. Define $\mathcal{E}_i$ to be the event that $2Crt_i \le \|p-x\| \le 2Crt_{i+1}$ (the range $2Crt_i$ to $2Crt_{i+1}$ are our `buckets' from the discussion preceding the proof). Then
\begin{equation}\label{evbound_soln}
   \E\|p-x\| = \sum_{i \ge -1} \E[\|p-x\| \mid \mathcal{E}_i] \Pr(\mathcal{E}_i). 
\end{equation}
We first bound $\Pr(\mathcal{E}_i)$ in two different ways. By conditioning on the value of $\tilde{r}$, we can write this probability as
\begin{align}
    \Pr(\mathcal{E}_i) &= \sum_{j \ge -1}\Pr(\mathcal{E}_i \text{ and } 2rt_j \le \tilde{r} \le 2rt_{j+1}) \label{eq:bound1} \\
    &= \sum_{j \ge -1}\big[\Pr(\mathcal{E}_i \mid 2rt_j \le \tilde{r} \le 2rt_{j+1}) \nonumber \\
    &\hspace{2cm}\cdot\Pr(2rt_j \le \tilde{r}\le 2rt_{j+1})\big] \label{eq:bound2}.
\end{align}
In the first of our two bounds for $\Pr(\mathcal{E}_i)$, we proceed by bounding $\Pr(2rt_j \le \tilde{r}\le 2rt_{j+1})$. Heuristically, the event $2rt_j \le \tilde{r}\le 2rt_{j+1}$ would mean that some point in $B(p, r)$ will be very far away from $p$ after the random projection and the probability of this event can be controlled very well. 

More formally, we first claim that the event $2rt_j \le \tilde{r}\le 2rt_{j+1}$ implies that there exists a point $z$ in $B(p, r)$ such that $\|G(z-p)\| \ge rt_j$. This is because otherwise, we have $\|Gp-Gq\| < rt_j$ for all $q \in B(p,r)$. This means that
\begin{align*}
    \sum_{q \in B(Gp, 2rt_j)} (2rt_j - \|Gp-Gq\| ) &> \sum_{q \in B(p, rt_j)} (2rt_j -rt_j) \\
    &\ge |B(p, r) \cap X| \cdot rt_j \\
    &\ge |B(p, r) \cap X| \cdot r.
\end{align*} 
We also know that $|B(p, r) \cap X| \cdot r \ge 1$ from \eqref{eq:numberlb}. Altogether, we have that $\sum_{q \in B(Gp, 2rt_j)} (2rt_j - \|Gp-Gq\| ) > 1$
which cannot happen by definition of $\tilde{r}$ and our assumption that $2rt_j \le \tilde{r}$ (see Figure \ref{fig:circle_radii}). Therefore by Lemma \ref{lem:indyk}, we have
\begin{equation}\label{eq:proof1}
     \Pr(2rt_j \le \tilde{r}\le 2rt_{j+1}) \le \exp(-C_1dt_j^2) 
\end{equation}
for some constant $C_1$.
Summing over the variable $j$ in inequality \eqref{eq:proof1} gives us a bound on $\Pr(\mathcal{E}_i)$. 
We will only end up using this bound for $j \ge \Omega(\sqrt{i})$, and will use the second bound for small $j$.

We now give a second bound on $\Pr(\mathcal{E}_i)$ by controlling
$\Pr(\mathcal{E}_i \text{ and } 2rt_j \le \tilde{r} \le 2rt_{j+1})$. Note that
the event $\mathcal{E}_i \text{ and } 2rt_j \le \tilde{r} \le 2rt_{j+1}$ together imply that there exists some $x$ that satisfies $$2Crt_i \le \|p-x\| \le 2Crt_{i+1} \ \text{ and } \|G(x-p)\| \le 2Crt_{j+1} $$
due to the fact that $Gx$ is a point in $B(Gp, C \tilde{r}_p)$. Therefore,
\begin{align*}
&\Pr(\mathcal{E}_i \text{ and } 2rt_j \le \tilde{r} \le 2rt_{j+1}) \\
\le &\Pr(\exists x, \, 2Crt_i \le \|p-x\| \le 2Crt_{i+1} \\
&\hspace{2cm} \text{ and } \|G(x-p)\| \le 2Crt_{j+1}).
\end{align*}
We bound the right hand side of the above probability for the range $j = O(\sqrt{i})$. Let $$X_i = \{x \in X \mid 2Crt_{i} \le \|x-p\| < 2Crt_{i+1} \}.$$ By the definition of doubling dimension, we can find a covering of $X_i$ with at most $\lambda^{O(\log(t_i/\epsilon))}$ balls of radius $ 2Cr\epsilon/4$ centered at points in some set $S \subseteq X$. Then by Lemma \ref{lem:indyk}, we have
\begin{multline}\label{eq:proof2}
    \hspace{-0.2cm}\Pr\bigg(\exists s \in S \ \exists x \in B(s, 2Cr\epsilon/4) \cap X_i, \\
    \hspace{1.2cm}\|Gs-Gx\| \ge \frac{2Cr\epsilon \sqrt{i}}{8}\bigg) \le \exp(-O(di))
\end{multline}
	if $d \ge \Omega(\log(\lambda) \log(1/\epsilon)/\epsilon^2)$. Now fix $s \in S$. If $\|G(s-p)\| < 2Cr(1 + \epsilon + \epsilon \sqrt{i}/4)$ then
	
\begin{align*}
    \frac{\|G(s-p)\|}{\|s-p\|} &\le \frac{1+\epsilon + \epsilon \sqrt{i}/4}{1+2\epsilon + \epsilon i/4} \\
    &\le 
	\begin{cases}
	1-\epsilon/4,  &\text{for} \ 0 \le i \le 1/\epsilon^2 \\
	O(1)/\sqrt{i}, &\text{for} \ i > 1/\epsilon^2 
	\end{cases}.
\end{align*}
Hence by applying the two inequalities \eqref{eq:Gxa} and \eqref{eq:Gxb} to the unit vector $(s-p)/\|s-p\|$, we have
\begin{align*}
    &\Pr\left(\exists s \in S, \|G(s-p)\| \le 2Cr\left(1 + \epsilon + \frac{\epsilon \sqrt{i}}{4} \right) \right) \\
	&\le
	\begin{cases}
     \exp(-c^{''}d \epsilon^2),  &\text{for} \ 0 \le i \le 1/\epsilon^2 \\
 i^{-c^{''}d}, &\text{for} \ i > 1/\epsilon^2 
	\end{cases}.
\end{align*}
	Note that we used the inequality \eqref{eq:Gxa} for the bound $i \le 1/\epsilon^2$ and the inequality \eqref{eq:Gxb} for $i > 1/\epsilon^2$. Combining the above bound with the inequality in \eqref{eq:proof2} gives us
\begin{align*}
    &\Pr\left(\exists x \in X_i, \|G(x-p)\| \le 2Cr\left(1+ \epsilon + \frac{\epsilon \sqrt{i}}{8}\right) \right) \\
    &\le 
	\begin{cases}
    2 \exp(-c^{''}d \epsilon^2),  &\text{for} \ 0 \le i \le 1/\epsilon^2 \\
 2i^{-c^{''}d}, &\text{for} \ i > 1/\epsilon^2 
	\end{cases}.
\end{align*}
	Thus for $j \le C_2\sqrt{i}$, we have
	\begin{align*}
	    &\Pr(\exists x, 2Crt_{j+1} > \|G(x-p)\| \\
	    &\hspace{2cm} \text{ and } 2Crt_{i} \le \|x-p\| < 2Crt_{i+1}) \\
	    &\le
	\begin{cases}
    2 \exp(-C_3d \epsilon^2),  &\text{for} \ 0 \le i \le 1/\epsilon^2 \\
 2i^{-C_3d}, &\text{for} \ i > 1/\epsilon^2 
	\end{cases}
	\end{align*}
	where $C_2, C_3$ are fixed constants. Using the representation given in \eqref{eq:bound2} for $\Pr(\mathcal{E}_i)$ along with \eqref{eq:proof1}, we see that for $0 \le i \le 1/\epsilon^2$, we can bound
	\begin{equation} \label{eq:prob_bound1}
	    	 \Pr(\mathcal{E}_i)  \le   4\exp(-C_2d \epsilon^2) + \sum_{j \ge 1} \exp(-C_1dj^2\epsilon^2)
	\end{equation}
	while for $i > 1/\epsilon^2$, we instead use the following stronger bound
	\begin{equation} \label{eq:prob_bound2}
	    	\Pr(\mathcal{E}_i) \le  2C_2\sqrt{i} \cdot  i^{-C_3d} + \sum_{j \ge C_2\sqrt{i}} \exp(-C_1dj^2\epsilon^2)
	\end{equation}
	which comes from using \eqref{eq:bound1} for $j \le C_2 \sqrt{i}$ and \eqref{eq:bound2} for larger $j$. Combing these bounds with \eqref{evbound_soln}, we have
	\begin{equation}\label{eq:expectation}
	    \E\|p-x\| \le 2C(1+O(\epsilon))r + \sum_{i \ge 0} 2Crt_{i+1}\Pr(\mathcal{E}_i). 
	\end{equation}
Our task is to now bound the sum $\sum_{i \ge 0} t_{i+1}\Pr(\mathcal{E}_i)$. In the rest of the proof, we will show that this sum is $O(\epsilon)$. We split the sum into two terms depending on if $i \le 1/\epsilon^2$ or if $i > 1/\epsilon^2$. Using the bounds 
\eqref{eq:prob_bound1} and \eqref{eq:prob_bound2} gives us
\begin{align}
    &\hspace{0.5cm}\sum_{i \ge 0}t_{i+1}\Pr(\mathcal{E}_i) \nonumber \\
    &\le C_4\sum_{i \le 1/\epsilon^2} i \bigg( \exp(-C_2 d \epsilon^2) +  \sum_{j \ge 1} \exp(-C_1dj^2\epsilon^2) \bigg) \label{eq:bound3} \\
    &+C_4 \sum_{i > 1/\epsilon^2}  \epsilon i \bigg( i^{-C_2 d + 1/2} +  \sum_{j \ge C_2\sqrt{i}} \exp(-C_1dj^2\epsilon^2) \bigg) \label{eq:bound4}
\end{align}
for some constant $C_4$. In \eqref{eq:bound3}, we are using the fact that $t_{i+1} = O(i)$ and for \eqref{eq:bound4}, we are instead using $t_{i+1} = O(\epsilon i)$. We can bound \eqref{eq:bound3}
\begin{equation}\label{eq:bound1b}
   \frac{\exp(-C_2 d \epsilon^2)}{\epsilon^4} + \frac{1}{\epsilon^2} \sum_{j \ge 1} \exp(-C_1 d j^2 \epsilon^2) \le O(\epsilon)
\end{equation}
by using the fact that $d = \Omega(\log(1/\epsilon)/\epsilon^2)$. 

We now focus on bounding \eqref{eq:bound4}. As a first step, we have the estimate
$$ \epsilon\sum_{i > 1/\epsilon^2}   i^{-C_2 d + 3/2} = O(\epsilon) $$
which holds for large enough constant $d$. Finally, bounding the remaining sum of \eqref{eq:bound4} by an integral gives us 
\begin{align*}
&\hspace{0.5cm} \epsilon \sum_{i > 1/\epsilon^2} i\sum_{j \ge C_2 \sqrt{i}}\exp(-C_1dj^2\epsilon^2) \\
&\le \epsilon \int_{1}^{\infty} x \int_{\sqrt{x}}^{\infty} \exp(-C_5d t^2 \epsilon^2) \ dt \ dx
\end{align*}
for some constant $C_5$. Now using the definition of the complementary error function, we can compute that
\begin{align*}
     &\hspace{0.5cm} \epsilon \int_{1}^{\infty} x \int_{\sqrt{x}}^{\infty} \exp(-C_5d t^2 \epsilon^2) \ dt \ dx  \\
     &\le O(d^{-1/2}) \int_1^{\infty} x \cdot \text{erfc}(\epsilon \sqrt{C_5d \cdot x}) \ dx  \\
     &\le O(\epsilon) \int_1^{\infty} x \cdot \text{erfc}(\epsilon \sqrt{C_5d \cdot x}) \ dx.
\end{align*}
From Lemma \ref{lem:tech}, we have 
\begin{multline}\label{eq:bound2b}
    \int_1^{\infty} x \cdot \text{erfc}(\epsilon \sqrt{C_5d \cdot x}) \ dx \le \int_1^{\infty}x \cdot \exp(-C_5^2\epsilon^2 d \cdot x) \ dx \\
    = O(\epsilon)
\end{multline}
using the fact that $d = \Omega(\log(1/\epsilon)/\epsilon^2)$. Altogether, the bounds \eqref{eq:bound1b} and \eqref{eq:bound2b} allow us to bound the right hand side of \eqref{eq:bound3} and \eqref{eq:bound4} and therefore, bound the sum $\sum_{i \ge 0} t_{i+1}\Pr(\mathcal{E}_i)$ as $O(\epsilon)$. Finally, using \eqref{eq:expectation}, we end up with
\begin{equation*}
   \E\|p-x\| \le 2C(1+O(\epsilon))r. \qedhere 
\end{equation*}
\end{proof}
As a corollary, we can prove Theorem \ref{cor:main}.
\begin{proof}[Proof of Theorem \ref{cor:main}]
Let $\mathcal{F}_d$ be a locally optimal solution in $GX$. When we evaluate the cost of $\mathcal{F}_d$ in the larger dimension $\mathbb{R}^m$, the number of facilities stays the same.  Now since $\mathcal{F}_d$ is a locally optimal solution in $\mathbb{R}^d$, each point $p$ has a facility that is within distance $C\tilde{r}_p$ in $\mathbb{R}^d$. Then by Theorem \ref{thm:main}, the connection cost of $p$ in the larger dimension is bounded by $C'r_p$, for some constant $C'$, in expected value. Summing over all points $p \in X$ gives us
$$\mathbb{E}[\cost_m(\F_d)] \le |\F_d|+O\bigg(\sum_{p \in X}r_p\bigg).$$
Finally, since $|\F_d| \le \cost_d(\F_d)$ by definition, and since $\sum_{p \in X} r_p = O(F)$ by Lemma \ref{lem:appx}, we have that
$$|\F_d|+O\bigg(\sum_{p \in X}r_p\bigg) \le \cost_d(\F_d)+O(F).$$
Together, these prove the theorem.
\end{proof}

\section{Dimension Reduction for MST: Omitted Proofs} \label{sec:app_mst}

In this section, we prove Lemma \ref{lem:MST} and Theorem \ref{thm:MainMST}.

\subsection{Proof of Theorem \ref{thm:MainMST}} \label{subsec:mst1}

In this subsection, we prove that Lemma \ref{lem:MST} implies Theorem \ref{thm:MainMST}. To see why, first note that $\cost_X(\widetilde{\mathcal{M}}) \ge \cost_X(\mathcal{M})$ and $\cost_{GX}(\mathcal{M}) \ge \cost_{GX}(\widetilde{\mathcal{M}})$, since $\mathcal{M}$ is the minimum spanning tree on $X$ and $\widetilde{\mathcal{M}}$ is the minimum spanning tree on $GX$. Moreover, for each edge $e = (x, y) \in \mathcal{M},$ $\|Gx-Gy\|$ has distribution $\chi_d/\sqrt{d} \cdot \|x-y\|,$ where $\chi_d$ is the square root of a chi-square with $d$ degrees of freedom. This has mean
\[\mu = \|x-y\| \cdot \frac{1}{\sqrt{2d}} \cdot \frac{\Gamma((d+1)/2)}{\Gamma(d/2)} = \|x-y\| \cdot \left(1 - O\left(\frac{1}{d}\right)\right)\]
and variance $\|x-y\|^2-\mu^2 = \|x-y\|^2 \cdot O(1/d)$ \cite{wolframchi}. Therefore, the standard deviation of $\|G(x-y)\|$ is at most $\epsilon \cdot \|x-y\|$ since $d = \Omega(\epsilon^{-2})$. Therefore, the expectation of $\cost_{GX}(\mathcal{M})$ is $\sum_{e = (x, y) \in \mathcal{M}} \|x-y\| \cdot (1-O(1/d)) = M \cdot (1-O(1/d))$. Also, using the well known fact that for any (possibly correlated) random variables $X_1, \dots, X_n,$ $\sqrt{Var(X_1+\dots+X_n)} \le \sum \sqrt{Var(X_i)},$ we have that the standard deviation of $\cost_{GX}(\mathcal{M})$ is at most $\sum_{e = (x, y) \in \mathcal{M}} \epsilon \cdot \|x-y\| = \epsilon \cdot M$.
    
To finish, define random variables $Z_1 = \cost_X(\widetilde{\mathcal{M}})-\cost_X(\mathcal{M}),$ $Z_2 = \cost_X(\mathcal{M}) - \cost_{GX}(\mathcal{M}),$ and $Z_3 = \cost_{GX}(\mathcal{M}) - \cost_{GX}(\widetilde{\mathcal{M}})$. Our observations from the previous paragraph tell us that $Z_1$ and $Z_3$ are nonnegative, and $Z_2$ has nonnegative expectation and standard deviation bounded by $O(\epsilon) \cdot M$. Finally, Lemma \ref{lem:MST} tells us that $\E[Z_1+Z_2+Z_3] \le O(\epsilon) \cdot M$. However, this means that $\E[Z_1] \le O(\epsilon) \cdot M,$ so $0 \le Z_1 \le O(\epsilon) \cdot M$ with high probability by Markov's inequality. Therefore, $\cost_X(\widetilde{\mathcal{M}}) \le (1+O(\epsilon)) \cdot M$ with high probability, so the pullback is a $1+O(\epsilon)$ approximation with high probability. Likewise, we also have that $0 \le Z_3 = O(\epsilon)$ with high probability, and since $\E[Z_2], \sqrt{Var(Z_2)} \le O(\epsilon) \cdot M$, we also have that $|Z_2| = O(\epsilon)$ with high probability. Thus, $|Z_2+Z_3| = O(\epsilon)$ with high probability, which means $\cost_{GX}(\widetilde{\mathcal{M}}) \in [1-O(\epsilon), 1+O(\epsilon)] \cdot M$. As a result, the MST cost is preserved under dimensionality reduction with high probability as well.

\subsection{Proof of Lemma \ref{lem:MST}} \label{subsec:mst2}

In this subsection, we prove prove Lemma \ref{lem:MST}. In fact, we show the following stronger statement.
\begin{multline} \label{Step1}
\hspace{-0.3cm}\E_G\left[\sum_{e = (x, y) \in \widetilde{\mathcal{M}}} \max(0, \|x-y\|-(1+5\epsilon)\|Gx-Gy\|)\right] \\
\le \epsilon \cdot M.
\end{multline}
To see why this implies Lemma \ref{lem:MST}, by removing the maximum with $0$, Equation \eqref{Step1} implies that $\E_G[\cost_{X}(\widetilde{\mathcal{M}}) - (1+5 \epsilon) \cdot \cost_{GX}(\widetilde{\mathcal{M}})] \le \epsilon \cdot M.$ But $\E_G[\cost_{GX}(\widetilde{\mathcal{M}})] \le \E_G[\cost_{GX}(\mathcal{M})] \le (1+\epsilon) \cdot M,$ which means that $\E_G[\cost_{X}(\widetilde{\mathcal{M}}) - \cost_{GX}(\widetilde{\mathcal{M}})] \le \epsilon \cdot M + 5 \epsilon \cdot (1+\epsilon) \cdot M = O(\epsilon) \cdot M.$

\begin{proof}[Proof of Equation \eqref{Step1}]
Consider some range $A_i = \left[(1+\epsilon)^i, (1+\epsilon)^{i+1}\right).$ We will bound the expectation of 
\[K_i := \sum_{\substack{e = (x, y) \in \widetilde{\mathcal{M}} \\ \|Gx-Gy\| \in A_i}} \max\big(0, \|x-y\|-(1+5\epsilon)\|Gx-Gy\|\big)\]
and sum our upper bounds for $\E_G[K_i]$ over a range of $i$. For $K_i$ to be nonzero, we need there to exist $(x, y)$ such that $\|x-y\| \ge \|Gx-Gy\|$ and $\|Gx-Gy\| \in A_i,$ so $\|x-y\| \ge (1+\epsilon)^i.$ Therefore, we only need to sum $\E[K_i]$ over integers $i$ such that $(1+\epsilon)^i \le \diam(X)$.

To do this, first consider some fixed $i$ and some sufficiently large constant $C_1$, and define $t := t_i := \frac{\epsilon}{C_1} \cdot (1+\epsilon)^i$. Consider the following greedy procedure of selecting a partition of $X$. First, pick some point $x_1$ arbitrarily, then pick some point $x_2$ of distance more than $t$ from $x_1$ (in the original space), then some point $x_3$ of distance more than $t$ from $x_1$ and $x_2$, and so on until we have some $x_1, \dots, x_r$ and can no longer pick any more points. Finally, we partition $X$ into subsets $X_1, \dots, X_r$ so that each $x \in X$ is in $X_p$ if $x_p$ is the closest point to $x$ (breaking ties arbitrarily). Note that the partitioning is deterministic (independent of $G$). We show the following proposition:

\begin{proposition} \label{MSTCircleBound}
    The MST cost $M$ of the dataset $X$ (in the original space $\mathbb{R}^m$) is at least $\frac{r \cdot t}{2}.$ 
\end{proposition}

\begin{proof}
By a known result on Steiner Trees \cite{steiner}, $M$ is at least $\frac{r}{2(r-1)}$ times the MST cost of the set $\{x_1, \dots, x_r\} \subset X$, assuming $r \ge 2$. As the distance between any $x_p, x_q$ is at least $r$, the MST cost of $\{x_1, \dots, x_r\}$ is at least $(r-1) \cdot t,$ so $M \ge \frac{r}{2(r-1)} \cdot (r-1) \cdot t = \frac{r \cdot t}{2}.$ Finally, as $(1+\epsilon)^i \le \diam(X)$, we have $t \le \diam(X)/C_1$, so if $C_1 > 2$, then the greedy procedure of partitioning $X$ cannot end with just $x_1$, so indeed $r \ge 2$.
\end{proof}

Now, we consider partitioning each $X_p$ into subsets $X_{p, 1}, \dots, X_{p, s}$ as follows. Since the radius of $X_p$ is at most $t$, by definition of the doubling dimension, for each $k \ge 1$ we can split $X_p$ into at most $\lambda_X^k$ balls of radius at most $t/2^k$. We choose the smallest integer $k$ so that all of these balls have diameter at most $2t$ when projected by $G$, and let $s_p = s$ be the number of subsets $X_{p, q}$ formed for each $p$. (Note: this partitioning $X_{p, q}$ is now dependent on $G$.) We claim the following:

\begin{proposition} \label{ReducedDoublingBound}
    For any fixed $p$ and all integers $k \ge 1,$ $\Pr(s_p > \lambda_X^{k}) \le \exp\left(-c d 2^{k}\right)$.
\end{proposition}

\begin{proof}
    For any fixed $k$, we split $X_p$ into at most $\lambda_X^k$ balls of radius at most $t/2^k$: this process is independent of $G$. Now, fix a small ball: when we apply the random projection $G$, the probability that it has radius more than $t$ when projected is at most $\exp\left(-c d 2^{2k}\right)$, by Lemma \ref{lem:indyk}. But there are $\lambda_X^k \le \exp\left(c d k\right)$ such balls if $d$ is at least $c^{-1} \log \lambda_X$, so the probability that even one of $G X_{p, q}$ has radius more than $t$ is at most $\exp\left(-c d 2^{2k}\right) \cdot \exp\left(c d k \right) = \exp\left(-c d (2^{2k}-k)\right) \le \exp\left(-c d 2^k\right)$.
\end{proof}

We also make the following observations:
\begin{enumerate}
    \item If $x \in X_p,$ then $\|x-x_p\| \le t$, so the diameter of each $X_p$ is at most $2t$. Likewise, the diameter of each $X_{p, q}$ is at most $2t$ in both the original space and the reduced space. \label{SmallRadius}
    %(in the original space) is at least $(r-1) \cdot t,$ since $x_1, \dots, x_r$ all have to be connected. If $(1+\epsilon)^i \le \diam(X) \cdot C \epsilon^{-1} \sqrt{\log n}$, then $t \le \diam(X)/10,$ so $r \ge 2,$ so $M \ge \frac{r \cdot t}{2}.$ 
    \item By properties of the doubling dimension, for any $x_p$ and all $k \ge 1,$ there are at most $\lambda_X^{C_2 \cdot k}$ points $\{x_{p'}\}_{p'= 1}^{r}$ within $2^k \cdot t$ of $x_p$  for some $C_2$, since $x_1, \dots, x_r$ are all at least $t$ apart. \label{DoublingBound}
    %\item For each $X_p$ and each $k \ge 1$, $\diam(GX_p) \ge 2^{k+1} \cdot t$ with probability at most $\exp\left(-c d 2^{2k}\right)$ by Lemma \ref{lem:indyk}. Therefore, for all $k \ge 1$, $\Pr(s_p \ge \lambda_X^{C_2 \cdot k}) \le \exp\left(-c d 2^{k}\right)$. \todo{Prove. Essentially follows from fact that we can split $X_p$ into $\lambda_X^k$ spheres of size $2^{-k}$ and we need at least one of these to blow up by a factor of $2^k$} \label{ReducedDoublingBound}
\end{enumerate}

Recall that $D(X_p, X_{p'})$ is the \emph{maximum} distance between points in $X_p$ and $X_{p'}$ (in the original space), as opposed to $d(X_p, X_{p'})$ which is the \emph{minimum} distance. Now, for any fixed $i$, we bound the expectation of
\[L_i := \sum_{\substack{p, p' \\ d(GX_p, GX_{p'}) < (1+\epsilon)^{i+1} \\ D(X_p, X_{p'}) \ge (1+5 \epsilon) \cdot (1+\epsilon)^i}} D(X_p, X_{p'}) \cdot s_p s_{p'},\]
where the sum is over all pairs $p, p' \in [r]$.

First, we make the following claim.

\begin{lemma}
    For all $i$ and any fixed $G$, $L_i \ge K_i$.
\end{lemma}

\begin{proof}
    For any edge $e \in \widetilde{\mathcal{M}},$ if $e$ has length in range $A_i$ (in the projected space), then this length is greater than $2C_1 \cdot t$ (assuming $\epsilon < 1/2$). Then, $e$ is some edge $(Gx, Gy)$ where $x \in X_{p, q}, y \in X_{p', q'}$, where $(p, q) \neq (p', q')$ by Observation \ref{SmallRadius}. So, if edge $e$ contributes toward the sum in $K_i,$ then $D(X_p, X_{p'}) \ge \|x-y\| \ge (1+5 \epsilon) \cdot \|Gx-Gy\| \ge (1+5 \epsilon) \cdot (1+\epsilon)^i.$ At the same time, $d(GX_p, GX_{p'}) \le \|Gx-Gy\| < (1+\epsilon)^{i+1}.$ Thus, this pair $(p, p')$ contributes toward the sum in $L_i$. Moreover, $D(X_p, X_{p'}) \ge \|x-y\| \ge \max(0, \|x-y\|-(1+5 \epsilon) \|Gx-Gy\|).$ This will be useful since $L_i$ is a sum over $D(X_p, X_{p'})$ (multiplied by $s_p s_{p'}$) and $K_i$ is a sum over $\max(0, \|x-y\|-(1+5 \epsilon) \|Gx-Gy\|)$.
    
    Finally, it is impossible for two pairs $(Gx, Gy)$ and $(Gx', Gy')$ to both be edges in $\widetilde{\mathcal{M}}$ that contribute to the sum $K_i$, if $x, x' \in X_{p, q}$ and $y, y' \in X_{p', q'}$. If there were such pairs $(Gx, Gy), (Gx', Gy')$, this means that the edges $(Gx, Gy)$ and $(Gx', Gy')$ have length in $A_i$, and therefore have length at least $2C_1 \cdot t$. However, the diameters of $G X_{p, q}$ and $G X_{p', q'}$ are at most $2 t$, so it would be better to replace edge $(Gx', Gy')$ with either edge $(Gx, Gx')$ or edge $(Gy, Gy')$: exactly one of these replacements will preserve the spanning tree property, and either replacement reduces the total cost. Thus, for each pair $(p, p')$ contributing to the sum in $L_i,$ at most $s_p \cdot s_{p'}$ corresponding pairs $(x, y)$ can contribute to the sum in $K_i$, and since $D(X_p, X_{p'}) \ge \max(0, \|x-y\|-(1+5 \epsilon) \|Gx-Gy\|)$ whenever $x \in X_{p, q}, y \in X_{p', q'}$, this finishes the proof.
\end{proof}

%Next, note that $s_p s_{p'} \le \frac{1}{2} (s_p^2+s_{p'}^2),$ so $D(X_p, X_{p'}) \cdot s_p s_{p'} \le \frac{1}{2} D(X_p, X_{p'}) \cdot (s_p^2+s_{p'}^2)$. Hence, by symmetry, we have that
%\[L_i \le \sum_{\substack{p, p' \\ d(GX_p, GX_{p'}) < (1+\epsilon)^{i+1} \\ D(X_p, X_{p'}) \ge (1+C \epsilon) \cdot (1+\epsilon)^i}} D(X_p, X_{p'}) \cdot \frac{s_p^2+s_{p'}^2}{2} = \sum_{\substack{p, p' \\ d(GX_p, GX_{p'}) < (1+\epsilon)^{i+1} \\ D(X_p, X_{p'}) \ge (1+C \epsilon) \cdot (1+\epsilon)^i}} D(X_p, X_{p'}) \cdot s_p^2 := L_i'.\]
%Hence, $K_i \le L_i'$. 
We will now bound the expectation of $L_i$.

\begin{lemma}
    For any fixed $i$, $\E[L_i] \le \frac{\epsilon^2}{10 \log n} \cdot M$.
\end{lemma}

\begin{proof}    
For each $j \ge 1,$ define $B_{i, j}$ to be the interval $[(1+5 \epsilon) \cdot (1+\epsilon)^{i+j-1}, (1+5 \epsilon) \cdot (1+\epsilon)^{i+j})$. Fix some $p, p'$ such that $D(X_p, X_{p'}) \in B_{i, j}$ (note: this is independent of $G$). Since all points in $X_p$ are at most $t$ from $x_p$ (and similar for $X_{p'}$), we have that $\|x_p-x_{p'}\| \ge (1+5 \epsilon) \cdot (1+\epsilon)^{i+j-1} - 2 t \ge (1+3 \epsilon) \cdot (1+\epsilon)^{i+j}.$ Now, if $d(GX_p, GX_{p'}) < (1+\epsilon)^{i+1},$ then one of the following three events must be true:
\begin{enumerate}
    \item $\|x_p - x_{p'}\| \le (1+\epsilon)^{i+(j/2)} \cdot (1+3\epsilon)$
    \item $\diam(GX_p) \ge \epsilon \cdot (1+\epsilon)^{i+(j/2)}$
    \item $\diam(GX_{p'}) \ge \epsilon \cdot (1+\epsilon)^{i+(j/2)}$.
\end{enumerate}
Indeed, if all three were false, then $d(GX_p, GX_{p'}) \ge \|x_p - x_{p'}\| - \diam(GX_p) - \diam(GX_{p'}) \ge (1+\epsilon)^{i+(j/2)} \cdot (1+\epsilon) \ge (1+\epsilon)^{i+1}$.

Now, the probability of the first event (over the randomness of $G$) is at most the probability that a random projection shrinks $x_p-x_{p'}$ by a factor of at least $(1+\epsilon)^{j/2}$. By Equation \eqref{eq:Gxa}, if $j \le \epsilon^{-1},$ then this happens with probability at most $\exp\left(-d(j \epsilon)^2/100\right)$, and by Equation \eqref{eq:Gxb}, if $j > \epsilon^{-1},$ then this happens with probability at most $(1+\epsilon)^{-(j/2) \cdot d/20} \le \exp\left(-d (j \epsilon)/100\right).$ The probability of each of the second and third events occurring, since $\diam(X_p), \diam(X_{p'}) \le \epsilon \cdot (1+\epsilon)^i/C_1,$ is at most $\exp\left(-c d \cdot C_1^2 (1+\epsilon)^{j}\right) \le \exp\left(-d (j \epsilon)/100\right)$ by Lemma \ref{lem:indyk}. Next, note that by Proposition \ref{ReducedDoublingBound}, for some constant $C_3,$ $\Pr(s_p \ge \lambda_X^{C_3 \cdot k}) \le \exp(-d \cdot 2^k/100)$ for all real $k \ge 1$, and the same is true for $s_{p'}$.

Again consider some fixed $j$ and some $p, p'$ with $D(X_p, X_{p'}) \in B_{i, j}.$ Define the random variable $Z_{p, p'} := s_p s_{p'} \cdot \mathbb{I}\left(d(GX_p, GX_{p'}) < (1+\epsilon)^{i+1}\right),$ where $\mathbb{I}$ represents an indicator random variable. Then, if $j \le \epsilon^{-1}$, $d(GX_p, GX_{p'}) < (1+\epsilon)^{i+1}$ occurs with probability at most $3 \cdot \exp\left(-d(j \epsilon)^2/100\right) \le 3 \cdot \exp\left(-d \cdot \epsilon^2/100\right)$, so $\Pr(Z_{p, p'} > 0) \le 3 \cdot \exp\left(-d \cdot \epsilon^2/100\right)$. Next, for any $k \ge 1$, if $Z_{p, p'} \ge \lambda_X^{2k \cdot C_3},$ then either $s_p$ or $s_{p'}$ is at least $\lambda_X^{k \cdot C_3}$, which occurs with probability at most $2 \exp\left(-d \cdot 2^k/100\right)$ by Proposition \ref{ReducedDoublingBound}. Hence,
\begin{align*}
\E[Z_{p, p'}] &\le 3 \cdot \exp\left(-\frac{d \epsilon^2}{100}\right) \cdot \lambda_X^{2 \cdot C_3} \\
&\hspace{1.5cm}+ \sum_{k = 1}^{\infty} 2 \cdot \exp\left(-\frac{d \cdot 2^k}{100}\right) \cdot \lambda_X^{2(k+1) \cdot C_3} \\
&\le 10 \cdot \exp\left(-\frac{d \epsilon^2}{200}\right)
\end{align*}
by our choice of the dimension $d$. However, if $j > \epsilon^{-1},$ then $d(GX_p, GX_{p'})$ occurs with probability at most $3 \cdot \exp\left(-d(j \epsilon)/100\right)$, so $\Pr(Z_{p, p'} > 0) \le 3 \cdot \exp\left(-d \cdot (j \epsilon)/100\right)$. But for any $k \ge 1,$ if $Z_{p, p'} \ge \lambda_X^{2(k+\log(j \epsilon)) \cdot C_3},$ then either $s_p$ or $s_{p'}$ is at least $\lambda_X^{(k+\log (j \epsilon)) \cdot C_3},$ which occurs with probability at most $2 \exp\left(-d \cdot 2^k \cdot (j \epsilon)/100\right).$ Hence,
\begin{align*}
    \E[Z_{p, p'}] &\le 3 \cdot \exp\left(-\frac{d(j \epsilon)}{100}\right) \cdot \lambda_X^{2(1+\log (j \epsilon)) \cdot C_3} \\
    &\hspace{0.5cm}+ \sum_{k = 1}^{\infty} 2 \cdot \exp\left(-\frac{d \cdot (j \epsilon) \cdot 2^k}{100}\right) \cdot \lambda_X^{2(k+1+\log (j \epsilon)) \cdot C_3} \\
    &\le 10 \cdot \exp\left(-\frac{d(j \epsilon)}{200}\right)
\end{align*}
    by our choice of the dimension $d$.
    
    Next, note that for each $p$, the number of $p'$ with $D(X_p, X_{p'}) \le (1+5 \epsilon) \cdot (1+\epsilon)^{i+j} \le \left[C_1 \cdot \epsilon^{-1} \cdot (1+\epsilon)^{5+j}\right] \cdot t$ is at most $\lambda_X^{C_2 \cdot (\log C_1 + \log \epsilon^{-1} + \epsilon \cdot (5+j))}$ by Observation \ref{DoublingBound}. Hence, the total number of pairs $(p, p')$ with $D(X_p, X_{p'}) \in B_{i, j}$ is at most $r \cdot \lambda_X^{C_2 \cdot (\log C_1 + \log \epsilon^{-1} + \epsilon \cdot (5+j))} \le r \cdot \lambda_X^{C_4 \cdot (\log \epsilon^{-1} + j \epsilon)}$ for some constant $C_4.$
    
    Combining everything together, we have that
\begin{align}
    \E[L_i] &= \sum_{j \ge 1} \sum_{p, p': D(X_p, X_{p'}) \in B_{i, j}} D(X_p, X_{p'}) \cdot \E[Z_{p, p'}] \nonumber \\
    &\le \sum_{j = 1}^{\epsilon^{-1}} \bigg(r \cdot \lambda_X^{C_4 \cdot (\log \epsilon^{-1} + j \epsilon)} \cdot (1+5 \epsilon) \cdot  (1+\epsilon)^{i+j} \nonumber \\
    &\hspace{2cm} \cdot 10 \cdot \exp\bigg(-\frac{d \epsilon^2}{200}\bigg)\bigg) \nonumber \\
    &+ \sum_{j > \epsilon^{-1}} \bigg(r \cdot \lambda_X^{C_4 \cdot (\log \epsilon^{-1} + j \epsilon)} \cdot (1+5 \epsilon) \cdot  (1+\epsilon)^{i+j} \nonumber \\
    &\hspace{2cm} \cdot 10 \cdot \exp\bigg(-\frac{d (j \epsilon)}{200}\bigg)\bigg) \nonumber \\
    &\le 20 C_1 r t \cdot \Biggr(\sum_{j = 1}^{\epsilon^{-1}} \lambda_X^{C_5 (\log \epsilon^{-1})} \exp\left(-\frac{d \epsilon^2}{200}\right)\\
    &\hspace{0.5cm}+\sum_{j > \epsilon^{-1}} \lambda_X^{C_5 (\log \epsilon^{-1} + j \epsilon)} \exp\left(-\frac{d (j \epsilon)}{200}\right)\Biggr) \label{bash} 
    %\\
    %&\le 40 C_1 \cdot M \cdot \left(\sum_{j = 1}^{\epsilon^{-1}} \lambda_X^{C_5 (\log \epsilon^{-1})} \exp\left(-\frac{d \epsilon^2}{200}\right)+\sum_{j > \epsilon^{-1}} \lambda_X^{C_5 (\log \epsilon^{-1} + j \epsilon)} \exp\left(-\frac{d (j \epsilon)}{200}\right)\right)
\end{align}
    for some constant $C_5$. Above, the first equality follows by definition of $L_i$. The next inequality follows from our bound on $\E[Z_{p, p'}]$, our bound the number of $(p, p')$ with $D(X_p, X_{p'}) \in B_{i, j},$ and since $D(X_p, X_{p'}) \in B_{i, j}$ implies $D(X_p, X_{p'}) \le (1+5 \epsilon) \cdot (1+\epsilon)^{i+j}.$ The final inequality follows from simple factorization and the facts that $C_1 t \le (1+\epsilon)^i$ and $1+5\epsilon \le 2$. %The last inequality follows by Observation \ref{MSTCircleBound}, which tells us that $rt \le 2 M$.
    
    Now, if we choose $d = C_6 \cdot (\log \log n + \log \epsilon^{-1} \log \lambda_X) \cdot \epsilon^{-2}$ for some sufficiently large constant $C_6,$ we have that $\lambda_X^{C_5 (\log \epsilon^{-1})} \cdot \exp(-d \epsilon^2/200) \le \frac{\epsilon^3}{1000 C_1 \log n}$ for all $j \le \epsilon^{-1}$, and $\lambda_X^{C_5(\log \epsilon^{-1}+j \epsilon)} \cdot \exp\left(-d(j \epsilon)/200\right) \le \exp\left(j \cdot \epsilon^{-1}\right) \cdot \exp\left(-d(j \epsilon)/400\right) \le \frac{\exp(-j)}{1000 C_1 \log n}$ for all $j > \epsilon^{-1}$. Hence, Equation \eqref{bash} can be upper bounded by
\begin{multline}
    20 C_1 \cdot rt \cdot \left(\epsilon^{-1} \cdot \frac{\epsilon^3}{1000 C_1 \log n} + \frac{\sum_{j \ge \epsilon^{-1}}e^{-j}}{1000 C_1 \log n}\right) \\
    \le \frac{\epsilon^2}{20 \log n} \cdot rt. \label{OverallBound}
\end{multline}
    Proposition \ref{MSTCircleBound} tells us that $rt \le 2 M$. Hence, Equation \eqref{OverallBound} is at most $\frac{\epsilon^2}{10 \log n} \cdot M,$ as desired.
\end{proof}

In sum, we have that $\E[K_i] \le \frac{\epsilon^2}{10 \log n} \cdot M$ for all $i$. Moreover, for small $i$, Equation \eqref{OverallBound} tells us that $\E[K_i] \le \frac{\epsilon^2}{20 \log n} \cdot rt \le \frac{\epsilon^2}{20 \log n} \cdot n \cdot (1+\epsilon)^i,$ since $r \le n$ and $t \le (1+\epsilon)^i$. Therefore, the LHS of Equation \eqref{Step1} is at most
\begin{multline*}
    \sum_{i: (1+\epsilon)^i \le \diam(X)} \E[K_i] \\
    \le \frac{\epsilon^2}{20 \log n} \cdot \sum_{i: (1+\epsilon)^i \le \diam(X)} \min\left(2 M, n \cdot (1+\epsilon)^i\right).
\end{multline*}
    Using the bound $\frac{\epsilon^2}{10 \log n} \cdot M$ for $i$ with $\frac{\diam(X)}{n} < (1+\epsilon)^i \le \diam(X)$ and the bound $\frac{\epsilon^2}{20 \log n} \cdot n \cdot (1+\epsilon)^i$ for $i$ with $(1+\epsilon)^i \le \frac{\diam(X)}{n},$ we can bound this by
\begin{multline*}
\frac{\epsilon^2}{20 \log n} \cdot \left(2 M \cdot \log_{1+\epsilon} n + \frac{\diam(X)}{1 - \frac{1}{1+\epsilon}}\right) \\
\le \frac{\epsilon \cdot M}{5} + \frac{\diam(X) \cdot \epsilon}{10 \log n} < \epsilon \cdot M,
\end{multline*}
    since $\diam(X) \le M.$ This concludes the proof.
\end{proof}

\section{Lower Bounds: Omitted Proofs} \label{sec:app_lb}

\subsection{Dependence on the Doubling Dimension}

In this subsection, we prove Theorems \ref{thm:lb_fac}, \ref{thm:lb_mst_cost}, and \ref{thm:lb_mst_pullback}.

We begin with Theorem \ref{thm:lb_fac}. To do so, we construct a set $X$ of $m$ points in $\mathbb{R}^m$ such that if we randomly project $X$ to $o(\log m)$ dimensions, then with high probability, the facility location cost is not preserved up to a constant factor. Moreover, the optimal set of facility centers in the projected space, with high probability, is not a constant-factor approximation to facility location in the original space. The point set $X$ we choose will just be a scaled set of identity vectors in $\mathbb{R}^m$: it is simple to see that this point set has $\lambda_X = m$. These points have the convenient property that each point's projection is independent of each other.

\begin{proof}[Proof of Theorem \ref{thm:lb_fac}]
    As mentioned previously, the points in $X$ will just be $R e_1, \dots, Re_m$, the $m$ identity unit vectors in $\mathbb{R}^m$ scaled by a factor $R \ge 1$. Since these points each have distance $R \sqrt{2} \ge \sqrt{2}$ from each other, the optimum set of facilities is all of them, which has cost $m$.
    
    Now, consider a random projection $G$ down to $d = o(\log m)$ dimensions, and define $C = \sqrt{\frac{\log n}{10 d}}$ and $R = \sqrt{C}$. Note that $R, C = \omega(1)$. Our goal will be to show that with at least $\frac{2}{3}$ probability, for all but $\frac{3m}{C}$ points $p \in GX$, $\tilde{r}_p \le \frac{2}{R}$, where we recall that $\tilde{r}_p$ is the positive real number such that
$$ \sum_{q \in B(p, \tilde{r}_p) \cap GX} (\tilde{r}_p - \|p-q\|) = 1.$$
    We trivially have the bound $\tilde{r}_p \le 1$ for all $p \in GX,$ which means that if we show our goal, then $\sum_{p \in GX} \tilde{r}_p \le m \cdot \frac{2}{R} + \frac{3m}{C} \cdot 1 \le \frac{5m}{R} = o(m)$. However, $\sum_{p \in GX} \tilde{r}_p$ is a constant-factor approximation to the optimum facility location cost by Lemma \ref{lem:appx}, which proves that the facility location cost of $GX$ is $o(m)$.
    
    Now, for each $e_i,$ by Equation \eqref{eq:Gxa}, we have for any $C \ge 6,$ 
\begin{align*}
    \Pr(\|Ge_i\| \le C) &\ge 1 - \exp\left(-d(C-1)^2/8\right) \\
    &\ge 1 - \exp\left(-(C-1)^2/8\right) \\
    &\ge 1 - \frac{1}{2C}.
\end{align*}
    Moreover, conditioned on $\|Ge_i\| \le C,$ by Lemma \ref{Nearby}, we have that for each $j \neq i,$ $\Pr(\|Ge_j-Ge_i\| \le \frac{1}{C}) \ge n^{-1/10}$ if $n$ is sufficiently large. Therefore, if $Ge_i: \|Ge_i\| \le C$ is fixed, since the $Ge_j$'s are independent vectors, we can apply the Chernoff bound to say that with probability at least $1-n^{-10},$ at least $\log n \ge R$ values of $j \neq i$ satisfy $\|Ge_j-Ge_i\| \le \frac{1}{C}$, or equivalently, $\|G(Re_j)-G(Re_i)\| \le \frac{R}{C} = \frac{1}{R}$. By removing our conditioning on $Ge_i$, we have that with probability at least $1 - \frac{1}{2C} - n^{-10} \ge 1 - \frac{1}{C},$ there are at least $R$ points in $GX$ that are within $\frac{1}{R}$ of $G(Re_i),$ in which case we have that for $p = G(Re_i)$, $\tilde{r}_p \le \frac{2}{R}.$ Therefore, in expectation, at most $\frac{m}{C}$ of the points in $GX$ have $\tilde{r}_p > \frac{2}{R}.$ Thus, by Markov's inequality, with probability at least $\frac{2}{3},$ at most $\frac{3m}{C}$ of the points in $GX$ have $\tilde{r}_p > \frac{2}{R}$. This proves the first part of the theorem.
    
    To prove the second part of the theorem, note that the optimal facility location cost over $GX$ is $o(m)$ with probability at least $2/3$, which implies that the number of open facilities in any optimal solution $\F_d$ is $o(m)$. But then, each point in $X$ which is not an open facility center is at least $R \sqrt{2}$ away from the nearest open facility center in the original space $X$, so the facility location cost in $X$ is at least $(m-o(m)) \cdot R \sqrt{2} = \omega(m)$.
\end{proof}

Next, we prove Theorems \ref{thm:lb_mst_cost} and \ref{thm:lb_mst_pullback}. These results prove that the dependence on doubling dimension $d_X$ is required in the projected dimension $d$, both to approximate the cost of the minimum spanning tree and to produce a minimum spanning tree in the lower dimension that is still an approximate MST in the original dimension.

\begin{proof}[Proof of Theorem \ref{thm:lb_mst_cost}]
    Let $X = \{0, e_1,\dots,e_m\}$, where $m = n-1,$ $0$ is the origin in $\mathbb{R}^m$ and $e_i$ is the $i$th identity vector for each $1 \le i \le m$. Clearly, the minimum spanning tree connects $0$ to all of the $e_i$'s and has cost $M = m$. Now, we show that for $C = \sqrt{\frac{\log n}{10d}} = \omega(1)$, the MST cost of $GX$ is at most $\frac{10 m}{C} = o(m)$ for sufficiently large $m$ with at least $2/3$ probability.
    
    To do so, note that since $G$'s entries are independent, $Ge_1, \dots, Ge_m$ are all i.i.d. $\frac{1}{\sqrt{d}} \cdot \mathcal{N}(0, I_d)$. Consider some $e_i, e_j$ and suppose that $\|G e_i\|, \|G e_j\| \le C$ but $\|G(e_i-e_j)\| \ge \frac{4}{C}.$ Then, if we let $v = G(e_i+e_j)/2,$ for each $k \neq i, j$, $\Pr(\|G e_k - v\| \le \frac{1}{C}) \ge n^{-1/10}$ by Lemma \ref{Nearby}. By the independence of $Ge_1, \dots, Ge_m$, with probability at least $1-n^{-10},$ there is some $k \neq i, j$ in $[m]$ such that $\|G e_k - v\| \le \frac{1}{C}$. In this case, the minimum spanning tree of $GX$ would not have the edge $(Ge_i, Ge_j),$ as this edge could be replaced by either the edge $(Ge_i, Ge_k)$ or $(Ge_k, Ge_j),$ both of which are shorter.
    
    Thus, with probability at least $1-n^{-8},$ if we just connect the points $Ge_i$ over all $i$ with $\|Ge_i\| \le C$ in an MST, every edge has length at most $\frac{4}{C}.$ We can create a possibly suboptimal spanning tree by connecting all $Ge_i$ with norm at most $C$ in an MST, connecting one of these vertices arbitrarily to $0 = G \cdot 0$, and finally connecting $Ge_i$ to $0$ for all $i$ with $\|Ge_i\| > C$. The first part has total cost at most $m \cdot \frac{4}{C}$ with probability at least $1-n^{-8}$. The second part has total cost at most $C$ with probability at least $1-n^{-8}$ (as long as some $\|Ge_i\| \le C$). Finally, the third part has total expected cost $m \cdot \E[\|Ge_i\| \cdot \mathbb{I}(\|Ge_i\| \ge C)],$ since each edge $e_i$ contributes to the third part only if $\|Ge_i\| \ge C$, and there are $m$ potential vertices $Ge_1, \dots, Ge_m.$ However, by the Cauchy-Schwarz inequality, we know that
\begin{align*}
\E\left[\|Ge_i\| \cdot \mathbb{I}(\|Ge_i\| \ge C)\right] &\le \sqrt{\E\left[\|Ge_i\|^2\right] \cdot \Pr(\|Ge_i\| \ge C)} \\
&\le \sqrt{1 \cdot \exp\left(-d \cdot (C-1)^2/8\right)} \\
&\le \exp\left(-(C-1)^2/16\right) \le \frac{1}{C},
\end{align*}
    with the final inequality true if $C \ge 7$. Therefore, with probability at least $\frac{4}{5},$ the third part has cost at most $\frac{5 m}{C}$ by Markov's inequality. So, with probability at least $\frac{4}{5} - 2 n^{-8} \ge \frac{2}{3},$ the total cost of this spanning tree in $GX$ (which may not even be minimal) is at most $\frac{4}{C} \cdot m + C + \frac{5}{C} \cdot m \le \frac{10 m}{C}$ assuming $m$ is sufficiently large.
\end{proof}

\begin{proof}[Proof of Theorem \ref{thm:lb_mst_pullback}]
    As in our proof of Theorem \ref{thm:lb_mst_cost}, let $C = \sqrt{\frac{\log n}{10 d}} = \omega(1)$. Consider $n = C \cdot m+1$ and let $X = \{0\} \cup \{e_i \cdot k/C\}$ for $1 \le i \le m, 1 \le k \le C$. The minimum spanning tree connects $0$ to $e_i/C$ to $2e_i/C$ to so on, so each edge has length $1/C$ and the total MST cost is $M = m$.
    
    Now, by Equations \eqref{eq:Gxa} and \eqref{eq:Gxb}, for each $e_i$, the probability that $\|Ge_i\| \in [1/10, 100]$ is at least $1 - \exp(-d/10) - (3/100)^d > 0.06$ for all $d \ge 1.$ Thus, with exponential failure probability in $m$, among $e_1, \dots, e_{m/2}$, at least $0.02 m$ of the $Ge_i$'s have norm between $1/10$ and $100$. Now, for some $i \le m/2$ with $1/10 \le \|Ge_i\| \le 100$, since $d = o(\log n)$, by Lemma \ref{Nearby}, the probability that $\|Ge_j-Ge_i\| \le \frac{1}{100 C}$ for any $j > m/2$ is at least $n^{-1/10}$. Hence, with exponential failure probability, for each $i$ with $1/10 \le \|Ge_i\| \le 100$, there is some $j > m/2$ with $\|Ge_j-Ge_i\| \le \frac{1}{100 C}.$
    
    Let $I$ be the set of $i$ such that $\|Ge_i\| \ge \frac{1}{10}$ and there is some $j$ with $\|Ge_j-Ge_i\| \le \frac{1}{100 C}.$ For each $i \in I,$ the distance between $Ge_i \cdot k/C$ and $Ge_i \cdot \ell/C$ for any $\ell \neq k$ is at least $\frac{1}{10C}$ but the distance between $Ge_i \cdot k/C$ and $Ge_j \cdot k/C$ is at most $\frac{1}{100C}.$ This means that the closest point to $Ge_i \cdot k/C$ in $GX$ is of the form $Ge_j \cdot k'/C$ for some $j \neq i$ and $k'$ which may or may not equal $k$. However, for every $Gx \in GX,$ the minimum spanning tree of $GX$ must contain the edge connecting $Gx$ to its closest neighbor, so for each $i \in I$ and $1 \le k \le C$, $\widetilde{\mathcal{M}}$ must connect $Ge_i \cdot k/C$ to $Ge_j \cdot k'/C$, which has length at least $k/C$ in the original space $\mathbb{R}^m$. Therefore, the pullback of the MST has length at least 
\[\sum_{i \in I} \sum_{k = 1}^{C} \frac{k}{C} \ge \frac{C}{2} \cdot |I|,\]
    which with exponential failure probability in $m$ is at least $\frac{C}{100} \cdot m = \frac{C}{100} \cdot M = \omega(M)$.
\end{proof}

\subsection{Approximate Solutions Cannot be Pulled Back}

In this subsection, we prove Lemmas \ref{lem:opt_nec} and \ref{lem:opt_nec_mst}. In other words, we give a simple example showing that our definition of \emph{locally optimal} (for FL) and that \emph{optimal} (for MST) is necessary, if we want dependence on $d_X = \log \lambda_X$ as opposed to $\log n$. 
In particular, our lemmas give examples showing that pulling back of \emph{any} approximately optimal solution found in the projected space to the original space does not work.
%Finally, we show that we must solve for an optimal solution in the reduced dimension, as opposed to an approximately optimal solution, if we want dependence on $\log \lambda_X$ as opposed to $\log n$.

\begin{proof}[Proof of Lemma \ref{lem:opt_nec}]
Consider the following set of points:
\[Y = \{b_1,b_2,  \ldots  , b_m\}= \{e_1, e_1 + e_2, \ldots, e_1 + e_2 + \cdots + e_m \} \]
where $e_i$ is the $i$th standard basis vector. We refer to this dataset as the `walk' dataset. Using the definition of doubling dimension (see Section \ref{sec:prelim}), we can compute that the doubling dimension of $Y$ is some constant independent of $m$. Now construct the dataset $X$ by scaling all the points in $Y$ by the factor $m^{1+1/2d}$. This does not affect the doubling dimension. Consider the projection of $X$ into $\mathbb{R}^d$ where $d = O(1)$. Before projection, the optimum solution is to open all facilities, costing $m$.

Now consider applying a random projection $G$ and note that the projection of the differences $G(b_i - b_{i+1})$ are independent. Therefore, by Proposition \ref{prop:ChiLB}, there is a pair of consecutive points $b_i, b_{i+1}$ such that $\|G(b_i-b_{i+1})\|$ shrinks by a factor of $C_1/m^{1/d}$ with probability at least $9/10$. Furthermore, by Equation \eqref{eq:Gxb}, we have that all the differences $\|G(b_i-b_{i+1})\|$ do not shrink by a factor worse than $C_2/m^{1/d}$ with probability at least $9/10$. Hence, with some constant probability, \emph{both} the following events occur:
\begin{itemize}
    \item There exists some $i^*$ such that $\|G(b_{i^*}-b_{i^*+1})\| = O(m^{1-1/2d})$
    \item $\|G(b_i-b_{i+1})\| = \Omega(m^{1-1/2d})$ for all $i$.
\end{itemize}
 In this case, the optimal solution in the projected space is to include all facilities, which has total cost $m$. However, a solution that is within a $1+O(m^{-1/2d})$ multiplicative factor of the optimal solution is to include all facilities except for $Gb_i^*$. However, evaluating this solution in the original dimension incurs a cost at least $\Omega(m^{1+1/2d})$, whereas the optimal cost is still $m$. Hence, the approach has approximation ratio of at least $m^{1/2d}$, which is $\omega(1)$, i.e., superconstant unless $d=\Omega(\log m)$.
\end{proof}

\begin{proof}[Proof of Lemma \ref{lem:opt_nec_mst}]
    Assume WLOG that $n = 2k^2$ for some $k$, that $X$ lies in $\mathbb{R}^{m}$ for $m = k+1$, and that $d = \epsilon \cdot \log n$ for some $\epsilon = o(1)$. Now, let $e_1, e_2, \dots, e_k$ represent the identity vectors in $\mathbb{R}^k$. Now, we will choose our $n$ points as follows. First, we will choose the $k^2$ points $X' = \{(0, \textbf{0}), (\frac{1}{k}, \textbf{0}), \dots, (\frac{k^2-1}{k}, \textbf{0})\},$ where $\textbf{0}$ represents the last $k$ coordinates all being $0$. For the remaining $k^2$ points, for each $0 \le i \le k-1$ we add the set $X_i = \{(i, e_i), (i + \frac{1}{k}, e_i) \dots, (i+\frac{k-1}{k}, e_i)\}$. We let $X = X' \cup X_0 \cup \dots \cup X_{k-1}$.
    
    First, we show that the doubling dimension of $X$, $\lambda_X,$ is at most $O(1)$. First, note that $X'$ and each $X_i$ is trivially embeddable into one dimension, because the points in $X'$ and in each $X_i$ only vary on one coordinate, so each of these individually have doubling dimension $O(1)$. Therefore, for any ball $B = B(r, p)$ of radius $r \le 10$ around some point $p$, $B \cap X$ is contained in some union of $O(1)$ of $X', X_0, \dots, X_{k-1}$. Consequently, the points in $B \cap X$ can be decomposed into $O(1)$ balls of radius $r/2$, since $B \cap X'$ and $B \cap X_i$ each have doubling dimension bounded by a constant. Now, if we consider some ball $B = B(r, p)$ of radius $r > 10,$ suppose that $p = (a_0, a_1, \dots, a_k) \in \mathbb{R}^{k+1}$. Now, consider the $5$ points $\{(a_0+\frac{j}{2} \cdot r, \textbf{0})\}_{j = -2}^{2},$ where the $\textbf{0}$ represents the last $k$ coordinates all being $0$. For every point $x$ in $X \cap B,$ $x$'s first coordinate must be in the range $[a_0-r, a_0+r]$ and $x$'s remaining coordinates have total magnitude at most $1$. With these two observations, it is immediate that every point in $X \cap B$ is within $r/2$ of some point $\{(a_0 + \frac{j}{2} \cdot r, \textbf{0})\}$ for some integer $-2 \le j \le 2$. Therefore, if $r > 10$, $B \cap X$ can be covered by $5$ balls of radius $r/2.$ Thus, $\lambda_X = O(1)$, so $X$ has doubling dimension $\log \lambda_X = O(1)$.
    
    Now, a straightforward verification tells us that for any $i \neq j,$ the points in $X_i$ and the points in $X_j$ are at least $\sqrt{2}$ away from each other. Moreover, each point $(i+\frac{j}{k}, e_i)$'s closest point in $X'$ is the corresponding point $(i+\frac{j}{k}, \textbf{0}),$ and this distance is $1$. Therefore, the minimum spanning trees of $X$ are as follows. First, connect the points in $X'$ in a line and all of the points in each $X_i$ in a line. Finally, for each $0 \le i \le k-1,$ choose some arbitrary $j$ and connect $(i+\frac{j}{k}, e_i)$ and $(i+\frac{j}{k}, \textbf{0}).$ The total MST cost $M$ is $\frac{k^2-1}{k} + k \cdot \frac{k-1}{k} + k \cdot 1 = 3 k - 1 - \frac{1}{k} = (3-o(1)) k$.
    
    Now, when the random projection $G: \mathbb{R}^{k+1} \to \mathbb{R}^d$ is applied, we have that each vector $(0, e_i)$ is independently mapped to some vector $(a_{i1}, \dots, a_{id})$, where each $a_{ij}$ for $1 \le i \le k, 1 \le j \le d$ is an i.i.d. $\mathcal{N}(0, 1/d)$. So for any $\epsilon = o(1)$ and $n$ sufficiently large, if we choose $\delta = e^{-1/(100 \epsilon)},$ we have that $\Pr(|a_{i1}|, \dots, |a_{id}| \le \delta/\sqrt{d}) = \Theta(\delta)^{d} \le e^{-\log n/4} < 1/\sqrt{2k},$ where we used the fact that $d = \eps \log n$. Hence, a simple Chernoff bound tells us that with $1-o(1)$ probability, at least $\sqrt{k}/2$ of the $(0, e_i)$'s get mapped to some $(a_{i1}, \dots, a_{id})$ with norm at most $\delta$.
    
    Now, consider the following $\omega(1)$-approximate MST for $X$. Let $A = \epsilon^{-1}$, and choose some set $I = \{i_1, \dots, i_A\}$. Our ``approximate'' MST will be as follows. For each $i \in I,$ remove the $k-1$ edges connecting $X_i$ together, and for each $1 \le j \le k,$ connect $(i + \frac{j}{k}, 0)$ with $(i + \frac{j}{k}, e_i)$. Each time this is done, we remove $k-1$ edges of length $1/k$ and add $k-1$ edges of length $1$ (recall that one of these edges of length $1$ was already in the MST), so the MST cost increases by $\epsilon^{-1} ((k-1) 1 - (k-1)/k) = \epsilon^{-1} k \cdot (1-o(1)).$ Hence, regardless of what set $A$ we chose, the approximate MST is a $\omega(1)$-approximation, as the true MST has cost $M = O(k)$.
    
    However, we claim that with high probability, we can choose $A$ so that this becomes a $(1+o(1))$-approximation in the projected space. Indeed, since $\epsilon \ge \frac{1}{\log n}$, with $1-o(1)$ probability, at least $\sqrt{k}/2 \ge \epsilon^{-1}$ values $e_i$ get mapped to some point with norm at most $\delta$. So, we choose $A$ to be of size $\epsilon^{-1}$ so that for all $i \in A,$ $e_i$ gets mapped to a point with norm at most $\delta$. Recall that $\mathcal{M}$ denote the true MST for $X$, and let $\mathcal{M}'$ be this poor-approximation spanning tree. Note that the only edges in $\mathcal{M}' \backslash \mathcal{M}$ connect $(i + \frac{j}{k}, 0)$ to $(i + \frac{j}{k}, e_i)$ for $i \in I, 0 \le j \le k-1$. Since there are $\eps^{-1} \cdot k$ such edges, and each edge has size at most $\delta$ when projected, we have that
\begin{align*}
    \cost_{GX}(\mathcal{M}') &\le \cost_{GX}(\mathcal{M}) + \delta \cdot \eps^{-1} \cdot k \\
    &\le \cost_{GX}(\mathcal{M}) + \eps^{-1} \cdot e^{-\eps^{-1}/100} \cdot k \\
    &= \cost_{GX}(\mathcal{M}) + o(k).
\end{align*}
    
    Now, let's suppose that $d \ge \omega(\log \log n)$. We saw in subsection \ref{subsec:mst1} that $\cost_{GX}(\mathcal{M})$ had expectation at most $M = \cost_X(\mathcal{M})$ and standard deviation $O(M/\sqrt{\log \log n})$, regardless of the dataset $X$. So, with $9/10$ probability, $\cost_{GX}(\mathcal{M}') \le \cost_{GX}(\mathcal{M}) + o(k) = (1+o(1)) M$. Moreover, by Theorem \ref{thm:MainMST}, with $9/10$ probability, $\widetilde{M},$ the cost of the MST in the reduced space $GX$, is within a $1 \pm o(1)$ factor of $M$. Therefore, with at least $4/5-o(1)$ probability, $\cost_{GX}(\mathcal{M}') \le (1+o(1)) \cdot \widetilde{M},$ so $\mathcal{M}'$ is an $\omega(1)$-approximate MST in $X$ but a $1+o(1)$-approximate MST in $GX$.
\end{proof}

\subsection{Lower Bounds for $k$-means and $k$-medians}

In this subsection, we prove Theorem \ref{thm:lb_kmeans}, which shows the tightness of the bounds of \cite{ilyapaper} for $k$-means and $k$-medians clustering even in the case of \emph{constant doubling dimension}.

We remark that \cite{ilyapaper} showed tightness of their result if doubling dimension is ignored. Namely, they showed the existence of such a point set $X$ that may have large doubling dimension. Hence, our contribution is making such a set that also has doubling dimension $O(1)$.

\begin{proof}[Proof of Theorem \ref{thm:lb_kmeans}]
    We start with the case where $n = 2t$ and $k = 2t-1$ for some $t$. As in \cite{ilyapaper}, we wish to consider $t$ pairs of points where each pair is of distance $1$ from each other, but all other distances are larger.
    
    Namely, we do the following. First, define $D = t^{1/d}/10,$ and let $R = \sqrt{D}$. We have that $D, R = \omega(1)$, since $d = o(\log n) = o(\log t)$. Now, for $1 \le i \le t,$ let $a_i = (2 \cdot i, \textbf{0}),$ meaning that $a_i$'s first coordinate is $2 \cdot i$ and the remaining $t = m-1$ coordinates are $0$. Next, for each $1 \le i \le t-1,$ define $b_i = a_i + e_{i+1}$, i.e., $b_i$ has first coordinate $2 \cdot i$, $(i+1)$th coordinate $1$, and all remaining coordinates $0.$ However, define $b_t = a_t + \frac{1}{R} \cdot e_{i+1}$. Our set $X$ will be the union of the $a_i$'s and $b_i$'s.
    
    Now, since $k = n-1,$ the $k$-medians cost of $X$ is just the distance between the closest pair of points in $X$, which is $\frac{1}{R}.$ The $k$-means cost of $X$ is just the squared distance between the closest pair of points in $X$. However, by Proposition \ref{prop:ChiLB}, for each $i$, 
\begin{align*}
    \Pr\left(\|G b_i - G a_i\| \le \frac{10}{t^{1/d}}\right) &= \Pr\left(\|G e_{i+1}\| \le \frac{10}{t^{1/d}}\right) \\
    &\ge \left(\frac{10}{e \cdot t^{1/d}}\right)^d \ge \frac{3}{t}.
\end{align*}    
    Moreover, since $e_2, \dots, e_{t}$ are all distinct unit vectors, the vectors $Ge_2, \dots, Ge_t$ are independent, which means that with probability at least $1 - (1 - 3/t)^{t-1} \ge 0.9$ (for $t$ sufficiently large), some $1 \le i \le t-1$ will have $\|Gb_i-Ga_i\| \le 10/t^{1/d} = 1/D.$ Thus, some pair of points $(a_i, b_i)$ satisfy $\|Ga_i-Gb_i\| \le 1/D$, whereas the closest distance between two points in $X$ was only $1/R$. Therefore, with at least $9/10$ probability, the $k$-medians cost has multiplied by a $R/D = o(1)$ factor after projection, and likewise, the $k$-means cost has multiplied by a $R^2/D^2 = o(1)$ factor.
    
    Now, let $p, q \in X$ be the pair of points minimizing $\|Gp-Gq\|$. With probability at least $4/5$, $\|Ga_t-Gb_t\| \ge 1/(20 R) > 1/D$, which means that either $p$ or $q$ is not in $\{a_t, b_t\}$: assume WLOG that $p \not\in \{a_t, b_t\}$. Thus, an optimal choice of $k$ centers (for either $k$-means or $k$-medians) is choosing all points in $X$, except $p$. But then, in the original space, these centers have $k$-medians cost equal to the distance from $p$ to its closest point in $X$, which is at least $1$. Likewise, the $k$-means cost is also at least $1$. However, the optimal $k$-medians and $k$-means costs are $1/R$ and $1/R^2$, respectively, so the optimal choice in $GX$ is an $R = \omega(1)$ or $R^2 = \omega(1)$ approximation for $k$-medians and $k$-means, respectively. This finishes the proof in the case that $k = n-1$.
    
    For general values of $k < n,$ we can simply consider having $n' = k+1$ points in the configuration as above, but with exactly one of the points replicated $n-k$ times. In this case, the cost of $k$-medians clustering is still the distance of the closest pair of distinct points, and the cost of $k$-medians clustering is still the square of the distance of the closest pair of distinct points. So, the lower bound of $\Omega(\log k)$ still holds.
\end{proof}

\section{Facility Location with Squared Costs} \label{sec:fl_squared}
%We consider the following variant of the facility location problem. 
Recall that the facility location with squared costs problem is defined as follows.
Given a dataset $X \subset \mathbb{R}^m$, our goal is to find a subset $\mathcal{F} \subseteq X$ that minimizes the objective
\begin{equation}\label{eq:objective2}
\cost(\mathcal{F}) =| \mathcal{F}| + \sum_{x \in X} \,  \min_{f \in \mathcal{F}}\|x-f\|^2. 
\end{equation}
%In contrast to \eqref{eq:objective}, we are adding the \emph{squared} distance from each point to its nearest facility in $\mathcal{F}$, rather than just the distance. 
%The difference between our new objective \eqref{eq:objective2} and the problem we considered in the previous sections which is given in \eqref{eq:objective}, can be compared to the difference between the $k$-median and the $k$-mean objectives. Therefore, we can think of \eqref{eq:objective2} as the `$k$-means' version of facility location clustering.

Similar to Equation \eqref{eq:rdef}, we give a geometric expression that is a constant factor approximation to the cost of the objective presented in \eqref{eq:objective2}. For each $p \in X$, associate it with a radius $r_p > 0$ that satisfies the relation
\begin{equation}\label{eq:rnewdef}
     \sum_{q \in B(p,r)} ( r_p^2 - \|p-q\|^2) = 1.
\end{equation}
We generalize the results in \cite{MP_alg} and \cite{sublin_MP} to give an analogue of Lemma \ref{lem:appx} for the squared objective \eqref{eq:objective2}.

\begin{lemma}\label{lem:appx2}
Let $C_{OPT}$ denote the cost of the optimal solution to the objective given in \eqref{eq:objective2}. Then
$$\frac{1}8 \cdot C_{OPT} \le \sum_{p \in X} r_p^2 \le 24 \cdot C_{OPT}. $$
\end{lemma}

To prove Lemma \ref{lem:appx2}, we first given an algorithm for \eqref{eq:objective2} inspired by the MP algorithm. Our algorithm, which we denote as the `Squared MP Algorithm,' is the following.
\\

\removelatexerror
\begin{algorithm}[H]
	\SetKwInOut{Input}{Input}
	\SetKwInOut{Output}{Output}
	\Input{Dataset $X = \{p_1, \cdots, p_n\} \subseteq \mathbb{R}^d$}
	\Output{Set $\mathcal{F}$ of facilities}
	\DontPrintSemicolon
	$\mathcal{F} \gets \emptyset$ \;
	\For{$i = 1$ to $n$}{
    Compute $r_i$ satisfying:
    	$ \sum_{q \in B(p_i, r_i)} (r_i^2 - \|p_i-q\|^2) = 1$
	\;
    }
    Sort such that $r_1 \le \ldots \le r_n$\;
    \For{$i=1$ to $n$}{
    \If{$B(p_i, 2r_i) \cap \mathcal{F} = \emptyset$}{
    $\mathcal{F} \gets \mathcal{F} \cup \{p_i\}$}
    }
    Output $\mathcal{F}$
	\caption{$\textsc{Squared MP Algorithm}$}
	\label{alg:MPalg_squared}
\end{algorithm}

We first claim that the set of facilities returned by Algorithm \ref{alg:MPalg_squared} is a constant factor approximation to the optimal set. 
\begin{theorem}\label{thm:MPalt}
Let $C_{OPT}$ denote the cost of the optimal solution to the objective given in \eqref{eq:objective2} and let $\F$ denote the set of facilities returned by Algorithm \ref{alg:MPalg_squared}. Then $\textup{cost}(\F) \le 6 \cdot C_{OPT}$.
\end{theorem}
\begin{proof} The proof follows similarly to Theorem $1$ in \cite{MP_alg} with some adaptations.
Let $\F'$ denote any set of facilities. For any point $x \in X$, let
$$ \text{charge}(x, \F') = d(x, \F')^2 + \sum_{p \in \F'} \max (0, r_p^2 - \|p-x\|^2)$$ where $d(x, \F')$ denotes the distance between $x$ and the closest point to $x$ in $\F'$ and $r_p$ is defined as in \eqref{eq:rnewdef}. We first show that $\sum_{x \in X} \text{charge}(x, \F') = \cost(F')$. Indeed, this follows from swapping the order of summation:
 \begin{align*}
    &\hspace{0.5cm} \sum_{x \in X} \text{charge}(x, \F') \\
    &=  \sum_{x \in X} \sum_{p \in \F'} \max (0, r_p^2 - \|p-x\|^2)  + \sum_{x \in X} d(x, \F')^2 \\
    &=  \sum_{p \in \F'} \sum_{x \in X} \max (0, r_p^2 - \|p-x\|^2)+ \sum_{x \in X} d(x, \F')^2  \\
    &= \sum_{p \in \F'} 1 + \sum_{x \in X} d(x, \F')^2  = \cost(\F').
\end{align*} 

Now denote $F^*$ as the set of facilities for the optimal solution. We first study the individual term $\text{charge}(x, \F^*)$. We first give a lower bound for $\text{charge}(x, \F^*)$. Let $q^*$ be the closest point to $x \in \F^*$. If $x \not \in B(q^*, r_{q^*})$ then $\text{charge}(x, \F^*) \ge \|x-q^*\|^2 > r_{q^*}^2$. Otherwise,
\begin{align*}
    \text{charge}(x, \F^*) &\ge \|x-q^*\|^2 + r_{q^*}^2 - \|x-q^*\|^2 \\
    &= r_{q^*}^2 \ge \|x-q^*\|^2 
\end{align*}  so altogether,
\begin{equation}\label{eq:charge1}
    \text{charge}(x, \F^*) \ge \max(r_{q^*}^2, \|x-q^*\|^2).
\end{equation}
Now let $\F$ denote the set of solutions returned by Algorithm \ref{alg:MPalg_squared}.
We now upper bound $\text{charge}(x, \F)$ in terms of the quantities $r_{q^*}^2, \|x-q^*\|^2$. Recall that $q^* \in \F^*$ is the closest point to $x$ in $\F^*$. We note that there must be a point $q \in \F$ such that $r_q \le r_{q^*}$ and $\|q-q^*\| \le 2r_{q*}$ due to how Algorithm \ref{alg:MPalg_squared} selects the set of facilities in step $6$. 

Now if $x \in B(q, r_q)$ then $d(x, \F) \le \|x-q\|$ and thus $ \text{charge}(x, \F) \le r_q^2$ since step $6$ of Algorithm \ref{alg:MPalg_squared} insures that $x \not \in B(q', r_{q'})$ for any other $q' \in \F$. Otherwise, $x \not \in B(q, r_q)$ in which case we claim that $\text{charge}(x, \F) \le \|x-q\|^2$. This claim is immediate unless there exists some $q' \in \F$ such that $x \in B(q', r_{q'})$. However in this case, a similar reasoning as above means $\text{charge}(x, \F) \le r_{q'}^2$ but 
\begin{equation*}
    \|x-q\| \ge \|q-q'\|-\|x-q'\| > 2r_{q'} - r_{q'} = r_{q'}
\end{equation*}
where the second inequality again follows from step $6$ of Algorithm \ref{alg:MPalg_squared}. Therefore, 
\begin{align}
    \text{charge}(x, \F) \le \|x-q\|^2 &\le (\|x-q^*\| + \|q^*-q\|)^2 \nonumber \\
    &\le 2\|x-q^*\|^2 + 2\|q^*-q\|^2 \nonumber \\
    &\le 2\|x-q^*\|^2 + 4r_{q^*}^2. \label{eq:charge2}
\end{align}
Comparing \eqref{eq:charge1} to \eqref{eq:charge2}, we can compute that the ratio of $2\|x-q^*\|^2 + 4r_{q^*}^2$ to $\max(r_{q^*}^2, \|x-q^*\|^2)$ is at most $6$ from which it follows that
 \begin{equation*}
    \text{charge}(x, \F) \le 6 \cdot  \text{charge}(x, \F^*).
\end{equation*}
Summing over $x \in X$ completes the proof.
\end{proof}

Using Theorem \ref{thm:MPalt}, we are now in position to prove Lemma \ref{lem:appx2}. The proof of Lemma \ref{lem:appx2} follows similarly to the proof of Lemma $2$ in \cite{sublin_MP} with some modifications to suit our alternate objective function given in \eqref{eq:objective2}.

\begin{proof}[Proof of Lemma \ref{lem:appx2}]
We first prove the lower bound. Note that for every $p_i \in X$, Algorithm \ref{alg:MPalg_squared} will open a facility within distance at most $2r_p$. Hence, $4 \sum_{p \in X} r_p^2$ is an upper bound on the cost to connect the points to their nearest facility. Now from similar reasoning as in the proof of Theorem \ref{thm:MPalt}, we note that each $p$ is in at most one ball $B(q, r_q)$ for some $q \in \F$, where $\F$ denotes the set of facilities returned by Algorithm \ref{alg:MPalg_squared}. Therefore,
\begin{equation*}
    \sum_{p \in X} r_p^2 \ge \sum_{q \in \F} \sum_{p \in B(q, r_q)} r_p^2.
\end{equation*}
Now if $p \in B(q, r_q)$ for some $q \in \F$ then we must have $r_q \le 2r_p$ because otherwise, step $6$ of Algorithm \ref{alg:MPalg_squared} would not have chosen $q$ as a facility center. Thus,
\begin{equation*}
\sum_{p \in X} r_p^2 \ge \sum_{q \in \F} \sum_{p \in B(q, r_q)} r_p^2 \ge \frac{1}4 \sum_{q \in \F} r_q^2 \cdot |B(q, r_q)|.
\end{equation*}
Finally, we know that
\[ 1 = \sum_{p \in B(q, r_q)}( r_q^2 - \|p-q\|^2) \le r_q^2 \cdot |B(q, r_q)| \]
from which it follows that $4\sum_{p \in X} r_p^2 \ge |\F|$. Altogether, we see that $8 \sum_{p \in X} r_p^2$ is an upper bound to the cost of the solution returned by Algorithm \ref{alg:MPalg_squared} so the lower bound follows.

For the upper bound, we will show that the sum of the radii squared is not too large compared to $\cost(\F)$ where $\F$ is the set of facilities returned by Algorithm \ref{alg:MPalg_squared}. Consider $p \not \in \F$ and let $q$ be the closest facility to $p$. First, we must have $r_p^2 \le 2(\|p-q\|^2 + r_q^2)$ because otherwise, $r_p^2 > (\|p-q\| + r_q)^2$ which implies that $B(q, r_q) \subseteq B(p, r_p)$. Furthermore,
\begin{align*}
   &\hspace{0.5cm}\sum_{p' \in B(p, r_p)} (r_p^2 - \|p-p'\|^2) \\
   &\ge \sum_{p' \in B(q, r_q)} (r_p^2 - \|p-p'\|^2) \\
   &> \sum_{p' \in B(q, r_q)}(2r_q^2 +2\|p-q\|^2 - \|p-p'\|^2) \\
   &\ge \sum_{p' \in B(q, r_q)}(r_q^2 +2\|p-q\|^2 + \|p'-q\|^2- \|p-p'\|^2) \\
   & \ge \sum_{p' \in B(q, r_q)} (r_q^2 - \|q-p'\|^2) = 1
\end{align*}
which contradicts \eqref{eq:rnewdef}. To summarize, if $p \not \in \F$ and $q$ is the closest facility in $\F$ to $p$, then 
\begin{equation}\label{eq:appx2}
    r_p^2 \le 2(\|p-q\|^2 + r_q^2).
\end{equation}
Going back to the upper bound, recall the definition of $\text{charge}(p, \F)$ used in the proof of Theorem \ref{thm:MPalt}:
$$ \text{charge}(p, \F) = d(p, \F)^2 + \sum_{q \in \F} \max (0, r_q^2 - \|q-p\|^2).$$
We also showed there that $\sum_{p \in X} \text{charge}(p, \F) = \cost(\F)$. Now
\begin{align*}
    \cost(\F) &= \sum_{p \in X} \text{charge}(p, \F) \\
    &\ge \sum_{q \in \F} r_q^2 + \sum_{p \in X \setminus \F} \max(r_{\delta(p)}^2, \|p-\delta(p)\|^2)
\end{align*}
where $\delta(p)$ denotes the closest element in $\F$ to $p$. From \eqref{eq:appx2}, we know that $r_p^2 \le 2(\|p-q\|^2 + r_q^2)$ so $\max(r_{\delta(p)}^2, \|p-\delta(p)\|^2) \ge r_p^2/4$ which gives us
\[ 6 \cdot C_{OPT} \ge  \cost(\F) \ge \frac{1}4 \cdot \sum_{p \in X} r_p^2, \]
as desired.
\end{proof}
We can prove the following statements about the expected value of $r_p$, defined as in \eqref{eq:rnewdef}, after a random projection to a suitable dimension depending on the doubling dimension of the set $X$. The following lemma is analogous to Lemmas \ref{lem:rlower} and \ref{lem:rupper} and omit its proof since the proof follows identically from the proofs in Lemmas \ref{lem:rlower} and \ref{lem:rupper}.

\begin{lemma}\label{lem:rnew_appx}
Let $X \subseteq \mathbb{R}^m$ and let $p \in X$. Let $G$ be a random projection from $\mathbb{R}^m$ to $\mathbb{R}^d$ for $d = O(\log \lambda_X)$. Let $r_p$ and $\tilde{r}_p$ be the radius of $p$ and $Gp$ in $\mathbb{R}^m$ and $\mathbb{R}^d$ respectively, computed according to Eq. \eqref{eq:rnewdef}. Then there exist constants $c, C > 0$ such that 
$$  cr_p^2 \le \E[\tilde{r}_p^2] \le C r_p^2.$$
\end{lemma}
Combining Lemma \ref{lem:rnew_appx}, which states that $\sum_p r_p^2$ is a constant factor approximation to thelb optimal solution of the objective given in \eqref{eq:objective2}, with Lemma \ref{lem:appx2}, we obtain the following theorem that is analogous to Theorem \ref{thm:constant}.
\begin{theorem}
Let $X \subseteq \mathbb{R}^m$ and let $p \in X$. Let $G$ be a random projection from $\mathbb{R}^m$ to $\mathbb{R}^d$ for $d = O(\log \lambda_X)$. Let $\mathcal{F}_m$ be the optimal solution in $\mathbb{R}^m$ and let $\F_d$ be the optimal solution for the dataset $GX \subseteq \mathbb{R}^d$. Then there exists constants $c,C>0$ such that
$$  c \cdot \textup{cost}(\F_m) \le \E[\textup{cost}(\F_d)] \le C \cdot \textup{cost}(\F_m).$$
\end{theorem}

Note that the crucial ingredient in the proof of Theorem \ref{thm:main} that allowed us to connect properties of the doubling dimension to facility location clustering was the relation given in Equation \eqref{eq:rdef}. The analogous relation for our new objective function in \eqref{eq:objective2} is given in \eqref{eq:rnewdef} and one can easily check that the steps in the proof of Theorem \ref{thm:main} transfer. Therefore, we have the following theorem.
\begin{theorem}\label{thm:main2}
Let $X\subseteq \mathbb{R}^m$ and let $G$ be a random projection from $\mathbb{R}^m$ to $\mathbb{R}^d$ for $d = O(\log \lambda_X \cdot  \log(1/\epsilon)/\epsilon^2)$. Fix $p \in X$ and let $Gx$ be any point in $B(Gp, C\tilde{r}_p)$ in $\mathbb{R}^d$ where $C$ is a fixed constant and $\tilde{r}_p$ is computed according to Eq. \eqref{eq:rnewdef} in $\mathbb{R}^d$. Then
$$ \E \|p-x\| \le 2C(1+O(\epsilon)) r_p. $$
\end{theorem}

To derive a statement analogous to Theorem \ref{cor:main} for our alternate objective function, we need a notion of a locally optimal solution. This task also follows from using Section \ref{sec:local} as a blue print. In particular, we can define local optimality of a solution to \eqref{eq:objective2} as follows.

\begin{definition}\label{def:local2}
A solution $\mathcal{F}$ to the objective given in \eqref{eq:objective2} is \emph{locally optimal} if for all $p \in X$, we have $B(p, 3r_p) \cap \F \ne \emptyset$ where $r_p$ is computed as in \eqref{eq:rnewdef}.
\end{definition}

Then the following lemma follows similarly to Lemma \ref{lem:localopt}.

\begin{lemma} \label{lem:localopt2}
Let $\mathcal{F}$ be an any collection of facilities. If there exists a $p \in X$ such that $B(p, 3r_p) \cap  \mathcal{F} = \emptyset$, then $ \textup{cost}(\mathcal{F} \cup \{p\}) < \textup{cost}(\mathcal{F})$,
i.e., we can improve the solution.
\end{lemma}

Finally, as a corollary to Lemma \ref{lem:localopt2} and Theorem \ref{thm:main2}, we have the following corollary.

\begin{corollary}\label{cor:main2}
Let $X \subset \mathbb{R}^m$ and let $G$ be a random projection from $\mathbb{R}^m$ to $\mathbb{R}^d$ for $d = O(\log \lambda_X \cdot  \log(1/\epsilon)/\epsilon^2)$. Let $\mathcal{F}_d$ be a locally optimal solution for the dataset $GX$ for the objective function given in \eqref{eq:objective2}. Then, the cost of $\mathcal{F}_d$ evaluated in $\mathbb{R}^m$, denoted as $\textup{cost}_m(\mathcal{F}_d)$, satisfies
$$ \E[\textup{cost}_m(\mathcal{F}_d)] \le |\mathcal{F}_d| + C' \cdot \sum_{p \in X} r_p$$
for some constant $C' > 0$.
\end{corollary}

\begin{remark}
We can compute that a constant smaller than $3$ works for Definition \ref{def:local2} and consequently Lemma \ref{lem:localopt2} but this choice is inconsequential since we already incur a multiplicative constant factor in Theorem \ref{thm:main2}.
\end{remark}

Finally, we argue that the lower bound of Theorem \ref{thm:lb_fac} also carries over to our new objective function, meaning that the dimension we project to must depend on the doubling dimension. We define the connection cost of the objective \eqref{eq:objective2} as the second portion.

\begin{theorem}\label{thm:lb2}
Let $d = o(\log n)$ and let $G$ be be a random projection from $\mathbb{R}^m$ to $\mathbb{R}^d$. There exists $X \subseteq \mathbb{R}^m$ where $|X| = n$ such that with at least $2/3$ probability, the optimal cost multiplies by $o(1)$ when projected. In addition, there exists an optimal solution $\widetilde{\mathcal{F}}$ in $\mathbb{R}^d$ that is only an $\omega(1)$-approximate solution in the original space $\mathbb{R}^m$.
\end{theorem}
\begin{proof}[Proof Sketch]
The proof follows similarly as in the proof of Theorem \ref{thm:lb_fac}. We again define $X = \{Re_1, \dots, Re_m\}$, where $R = \sqrt{C}$ and $C = \sqrt{\frac{\log n}{10 d}}$. As in the proof of Theorem 6.1, we again have for any fixed $p = R e_i$, with probability at least $1 - \frac{1}{C},$ there are at least $R$ points in $GX$ within $\frac{1}{R}$ distance of $Gp$. For any such point $p$, letting $\tilde{r}_p$ be the associated radius for $GX$ around $Gp$ as computed by Equation \eqref{eq:objective2}, we have that $\tilde{r}_p \le \frac{2}{\sqrt{R}}$. So, with at least $2/3$ probability, at most $\frac{3m}{C}$ of the points have $\tilde{r}_p > \frac{2}{\sqrt{R}} = o(1)$. As in the proof of Theorem 6.1, this shows that the optimal cost multiplies by a $o(1)$ factor, by using Lemma \ref{lem:appx2} this time.

In the original space $X \subset \mathbb{R}^m$, the optimal squared facility location cost is $m$, which is achievable by setting every point in $X$ as a facility. However, since the optimal facility cost in $GX$ is $o(m),$ the optimal solution $\widetilde{\mathcal{F}}$ in the reduced space $\mathbb{R}^d$ assigns at most $o(m)$ points to be facilities. Therefore, for the remaining $m-o(m)$ points, the connection cost in the original space is at least $(R \sqrt{2})^2 \ge R^2,$ so the cost of $\widetilde{\mathcal{F}}$ in the original space $X$ is at least $R^2 \cdot (m-o(m)) = \omega(1) \cdot m$. Thus, any optimal solution $\widetilde{\mathcal{F}}$ is an $\omega(1)$-approximate solution in the original space $\mathbb{R}^m$.
\end{proof}

\end{document}